\documentclass[a4paper,onecolumn]{quantumarticle}
\pdfoutput=1

\usepackage[utf8]{inputenc}
\usepackage{amsmath}
\usepackage{amssymb}
\usepackage{amsthm}
\usepackage{mathrsfs}
\usepackage{xcolor}
\usepackage{braket}
\usepackage{hyperref}

\theoremstyle{plain}
\newtheorem{defi}{Definition}
\newtheorem{lemma}[defi]{Lemma}

\newtheorem{prop}[defi]{Proposition}
\newtheorem{theo}[defi]{Theorem}
\theoremstyle{definition}

\def\hh{\frak h}
\def\rr{\mathbb{R}}
\def\cc{\mathbb{C}}
\def\nn{\mathbb{N}}
\def\PP{\mathbb{P}}
\def\MM{M_d(\cc)}
\def\tr{{\rm tr}}
\def\TR{{\rm TR}}

\def\Id{{\rm Id}}

\begin{document}

\title{Concentration Inequalities for Output Statistics of Quantum Markov Processes}
\author[1]{Federico Girotti}

\email{federico.girotti@nottingham.ac.uk}

\author[2,3]{Juan P. Garrahan}
\author[1,3]{M\u{a}d\u{a}lin Gu\c{t}\u{a}}

\affil[1]{School of Mathematical Sciences, University of Nottingham, Nottingham, NG7 2RD, United Kingdom}
\affil[2]{School of Physics and Astronomy, University of Nottingham, Nottingham, NG7 2RD, United Kingdom}
\affil[3]{Centre for the Mathematics and Theoretical Physics of Quantum Non-equilibrium Systems, University of Nottingham, Nottingham, NG7 2RD, United Kingdom}

\maketitle
\begin{abstract}
       We derive new concentration bounds 
    for time averages of measurement outcomes in quantum Markov processes. 
This generalizes well-known bounds for classical Markov chains which provide constraints on finite time fluctuations of time-additive quantities around their averages. We employ spectral, perturbation and martingale techniques, together with noncommutative $L_2$ theory, to derive:
    (i) a Bernstein-type concentration bound for time averages of the measurement outcomes of a quantum Markov chain, (ii) a Hoeffding-type concentration bound for the same process, (iii) a generalization of the Bernstein-type concentration bound for counting processes of continuous time quantum Markov processes, (iv) new concentration bounds for empirical fluxes of classical Markov chains which broaden the range of applicability of the corresponding classical bounds beyond empirical averages. 
    We also suggest potential application of our results to parameter estimation and consider extensions to reducible quantum channels, multi-time statistics and time-dependent measurements, and comment on the connection to so-called thermodynamic uncertainty relations.
    \end{abstract}

\section{Introduction}
Quantum Markov chains 
describe the evolution of a quantum system which interacts successively with a sequence of identically prepared ancillary systems (input probes) modelling a memoryless environment  \cite{Accardi,Kummerer2002}. 
Such dynamics can be seen as discrete-time versions of continuous-time open system evolutions as encountered in quantum optics \cite{Gardiner2004,Carmichael,Plenio1998}, and formalised the input-output theory of quantum filtering and control \cite{Wiseman,Bouten,Belavkin,Petersen,Gough}. The two settings are in fact closely connected and can be related explicitly by time discretising techniques \cite{AP06}.


After the interaction, the probes (output) are in a finitely correlated state \cite{FNW92} which carries information about the dynamics; such information can be extracted by performing successive measurements on the outgoing probes. The stochastic process given by the sequence of measurement outcomes has received significant attention in recent years. The ergodic properties of this process have been studied in \cite{KM03,AGS15,CP15,GH15,CGH21}, while those of the corresponding conditional system state (filter) have been analysed in \cite{KM04,BPT17,BFPP19,ABP21}. The large deviations theory of additive functionals of the measurement process was established in 
\cite{Mosonyi} and extended in \cite{Ogata}. Motivated by the interest in understanding the irreversible essence of repeated quantum measurements from a statistical mechanics perspective, the papers \cite{BJPP18,BCJP21} have investigated large time asymptotics of the entropy production. In the context of open quantum walks \cite{APSS12}, the asymptotics of the process obtained measuring the position of the walker on the graph have been studied \cite{AGS15,CP15,BBP17,PS19,CGH21}. 
Another class of problems relates to fluctuations of measurement outcomes, which can be used to uncover dynamical phase transitions of the quantum evolution, see for example \cite{Garrahan2010,Hickey2012,GH15}. In the one-atom maser, non-demolition measurements of outgoing atoms have been used for reconstructing the initial photon distribution of a resonant electromagnetic cavity \cite{Gal07}. Further theoretical studies have been carried out in \cite{BBB13, BB11}.


In contrast to the recent progress in understanding the \emph{asymptotics}
of the outcome process, much less is known on the \emph{finite-time} properties. Notable recent results in this direction are the deviation bounds and concentration inequalities for quantum counting processes and homodyne measurements obtained in \cite{BR21} (hence regarding continuous time models). The main aim of the present work is to provide new classes of concentration inequalities for additive functionals of the finite-time measurement process, 
which complement the results of \cite{BR21} and offer 
more explicit bounds. 
We denote by $X_1,\dots,X_n$ the outcomes of the measurements on the first $n$ output probes and we consider a generic function $f:I \rightarrow \rr$ defined on the set $I$ of all possible outcomes (we will mainly consider the situation in which we perform the same measurement on every probe, treating the general case in Subsection \ref{ss:timedep}); under certain ergodicity conditions on the dynamic, the empirical average process $\frac{1}{n}\sum_{k=1}^{n} f(X_k)$ converges almost surely to its stationary value $\pi(f)$ and we aim at finding a upper bound for the probability that $\frac{1}{n}\sum_{k=1}^{n} f(X_k)$ deviates from $\pi(f)$ more than $\gamma$ for some fixed $n \in \nn$ and $\gamma>0$. 

To get the general flavour of the results, we briefly discuss here the case of independent random variables for which there is a well established theory even for more general functions than the time average \cite{BLM13}; two well known bounds in this case are given by Bernstein's and Hoeffding's inequalities. Assuming for simplicity that $(X_k)_{k=1}^{n}$ are independent and identically distributed centered random variables, and that $\mathbb{E}[X^2_1]\leq b^2$ and $|X_1| \leq c$ almost surely, Bernstein's inequality reads as follows (see \cite[Theorem 2.10]{BLM13}):
\[\PP \left (\frac{1}{n} \sum_{k=1}^n X_k \geq \gamma \right ) \leq \exp \left (-n \frac{\gamma^2 b^2}{c^2} h \left ( \frac{2\gamma c}{b^2} \right ) \right )
\]
where $h(x)=(\sqrt{1+x}+x/2+1)^{-1}$; the bound depends on the random variables through their variance an their magnitude in absolute value. On the other hand, Hoeffding's inequality only depends on the extension of the range of $X_1$: if we have that $a \leq X_1 \leq b$ almost surely, one can show (\cite[Theorem 2.8]{BLM13}) that
\[\PP \left (\frac{1}{n} \sum_{k=1}^n X_k \geq \gamma \right ) \leq \exp \left (-n \frac{2 \gamma^2}{(b-a)^2} \right ).
\]
Much work has been done for extending these results to irreducible Markov chains. New Bernstein and Hoeffding-type inequalities were recently obtained in order to tackle problems coming from statistics and machine learning \cite{BQJ18, FJQ21}, while Lezaud \cite{Le98} and Glynn et al.
\cite{GO02} derived bounds which have an intuitive interpretation and depend on relatively simple dynamical properties. More specifically, the latter bounds involve the range and stationary variance of the function $f$, and the spectral gap of the (multiplicative symmetrization of the) transition matrix $P$ of the Markov chain (Bernstein case) or the norm of the pseudo-resolvent $({\rm Id}-P)^{-1}$ (Hoeffding case).


In this paper we find analogous bounds for the class of stochastic processes given by repeated measurements on the output probes of a quantum Markov process. On a technical level, we exploit perturbation theory and spectral methods used in \cite{CP15,GH15,CGH21} 
for proving the law of large numbers, central limit theorem and large deviation principle for the output process. We also use a generalisation of Poisson's equation to decompose the empirical average into a martingale with bounded increments and a negligible reminder as in \cite{AGS15}. Using these tools we obtain quantum bounds which share the useful properties of the classical results in \cite{Le98} and \cite{GO02} .


Our main results are: 
\begin{itemize}
    \item A Bernstein-type concentration bound for time averages of the measurement outcomes of a quantum Markov process, {\bf Theorem \ref{theo:main1}} below.
    \item A Hoeffding-type concentration for the same process, {\bf Theorem \ref{theo:main2}}.
    \item A generalization of the Bernstein-type concentration bound for counting processes of continuous-time quantum Markov processes, {\bf Theorem \ref{thm:ctbound}}.
    \item By specialising to classical Markov chains, we obtain a new concentration bound for empirical fluxes, thus extending the range of applicability of the corresponding classical bounds beyond empirical averages, {\bf Proposition \ref{prop1}} and {\bf Proposition \ref{prop2}}.
\end{itemize}

The paper is organised as follows. In Section ~\ref{sec:n&p} we introduce in more detail the mathematical model and we recall the objects and results that will be used to prove the main theorems. Sec. \ref{sec:QMC} is devoted to prove the main results of the paper: a Bernstein-type and a Hoeffding-type inequalities for the output process of quantum Markov chains. In Sec.~\ref{sec:QCP} we show how perturbation theory and spectral techniques can be used to derive concentration inequalities for the case of quantum counting processes too. This result integrates the bounds obtained in \cite{BR21}, providing a simple bound also for the case of counting processes and non-selfadjoint generators. Moreover, it bypasses the problem of establishing functional inequalities and estimating the constants appearing in the inequalities. In Sec.~\ref{sec:e&a} we present extensions and applications of the previous results. Finally, in Sec.~\ref{sec:conclusions} we provide our conclusions and outlook.

\section{Notation and preliminaries} \label{sec:n&p}

\subsection{Quantum channels and irreducibility}

We consider a finite dimensional Hilbert space $\cc^d$ and we denote by $\MM$ the set of $d\times d$ matrices with complex entries. When considering the evolution of a quantum system described by $\cc^d$ in the Schr{\"o}dinger picture, one endows $M_d(\cc)$ with the trace norm, i.e.
\[\|x\|_1:= \tr(|x|), \quad \forall x \in \MM,
\]
where we recall that $|x|$ is the unique positive semidefinite square root of $x^*x$ in the sense of functional calculus.
A state is any positive semidefinite $x \in \MM$ with unit trace. In the Heisenberg picture, $\MM$ is considered together with the uniform norm,
that is
\[\|x\|:=\sup_{u \in \cc^d, \|u\|=1}\|xu\|, \quad \forall x \in \MM.
\]
An observable is any selfadjoint element $x \in \MM$. We recall that the dual of $\MM$ considered with the uniform norm can be identified with $\MM$ with the trace norm via the following isometry:
\[(\MM,\| \cdot \|_1) \ni x \mapsto \tr(\cdot x) \in (\MM,\|\cdot \|)^*.
\]
We consider a completely positive unital 
linear map (that is, a quantum channel in the Heisenberg picture) 
$\Phi: \MM\to \MM$ and we denote by $\Phi^*$ the dual of $\Phi$, which is the completely positive and trace preserving 
(that is, stochastic) 
map uniquely defined by the relation:
\[\tr(x\Phi(y))=\tr(\Phi^*(x)y), \quad \forall x,y \in \MM.
\]
We recall that every completely positive map $\eta$ admits a Kraus representation $\eta(x)=\sum_{j}W_j^*xW_j$, where $\{W_j\} \subseteq \MM$ is a finite collection of operators  (called Kraus operators), and $\eta$ is a quantum channel if and only if $\sum_{j}W^*_jW_j=\mathbf{1}$. Unless stated otherwise, throughout this paper we will make the following assumption. 
\vspace{2mm}

\noindent
{\bf Hypothesis (H):} $\Phi$ is such that its dual admits a unique faithful invariant state $\sigma$, 
that is $\Phi^*(\sigma)=\sigma$ and $\sigma>0$.
\vspace{2mm}

We recall that hypothesis (H) is satisfied if and only
if the following equivalent statements hold:
\begin{enumerate}
    \item $1$ is an algebraically simple eigenvalue of $\Phi$ with positive eigenvector,
    \item let $\{V_i\}_{i \in I}$ be a choice of Kraus operators for $\Phi$, then for every $v \in \cc^d$
    \[{\rm span}\{V_{i_n} \cdots V_{i_1}v: n \in \mathbb{N}, i_1,\dots i_n \in I\}=\cc^d.
    \]
\end{enumerate}
In this case the map $\Phi$ is said to be irreducible. The equivalence is a special instance of Perron-Frobenius theory for completely positive (not necessarily unital)  maps and that we report in the following proposition (for more details see \cite[Chapter 6]{Wo12} and \cite[Section 3]{CP15}).
\begin{prop} \label{prop:PF}
Let $\eta$ be a completely positive map acting on $\MM$ for some finite dimensional Hilbert space $\cc^d$, then its spectral radius $r(\eta):=\sup\{|z|:z \in {\rm Sp}(\eta)\}$ is an eigenvalue of $\eta$ with positive semidefinite corresponding eigenvector $x \geq 0$. Moreover, the following are equivalent:
\begin{enumerate}
    \item $r(\eta)$ is algebraically simple and $x>0$,
    \item let $\{W_j\}_{j \in J}$ be a choice of Kraus operators for $\eta$, then for every $v \in \cc^d$
    \[{\rm span}\{W_{j_n} \cdots W_{j_1}v: n \in \mathbb{N}, j_1,\dots j_n \in J\}=\cc^d.
    \]
\end{enumerate}
\end{prop}

A stronger assumption which will not be required for our results, but provides a clearer picture of the dynamical aspects is primitivity. The channel $\Phi$ is called primitive if it satisfies hypothesis (H) and in addition it is aperiodic, i.e. its peripheral spectrum (the set of eigenvalues with absolute value $1$) contains only the eigenvalue $1$. 

\subsection{KMS-inner product}

\sloppy There are several ways of equipping $\MM$ with a Hilbert space structure, see for instance \cite{AC21} and \cite{PG11} for connections with quantum statistics. In the derivation of our results we will make use of the Kubo–Martin–Schwinger (KMS) inner product associated to a positive definite state $\sigma$, which is defined as follows 
\begin{equation}
    \langle x,y \rangle:=\tr((\sigma^{\frac{1}{4}} x \sigma^{\frac{1}{4}})^* (\sigma^{\frac{1}{4}}y\sigma^{\frac{1}{4}}))=\tr(\sigma^{\frac{1}{2}}x^* \sigma^{\frac{1}{2}}y), \quad x,y \in \MM.
\end{equation}
As usual, we write $\|x\|_2$ for $\langle x,x \rangle^{1/2}$. Given any map $\eta$ acting on $\MM$, we denote $\TR(\eta)$ the trace of $\eta$, that is the unique linear functional on linear maps acting on $\MM$ which is cyclic and such that $\TR({\rm Id})=\dim(\MM)=d^2$; we recall that, given any orthonormal basis $\{x_j\}_{j=1}^{d^2}$ of $\MM$ with respect to the KMS-product, the usual formula for computing the trace as the sum of diagonal elements holds true:
\[\TR(\eta)=\sum_{j=1}^{d^2} \langle x_j,\eta(x_j) \rangle.
\]
We will use the notation $\eta^\dagger$ for referring to the adjoint with respect to the KMS-product of the linear map $\eta$ acting on $\MM$; it is easy to see that $\eta^\dagger= \Gamma_\sigma^{-\frac{1}{2}} \circ \eta^* \circ \Gamma_\sigma^{\frac{1}{2}}$, where $\Gamma^\alpha_\sigma$ is the completely positive map defined as $x\mapsto\sigma^\alpha x \sigma^\alpha$, for $\alpha \in \rr$. If $\eta$ is completely positive, so does $\eta^\dagger$ and this is what motivates our choice of inner product, since the KMS-product is the only one with this property. Moreover, a Kraus decomposition of $\eta$ induces a Kraus decomposition of $\eta^\dagger$: if $\eta(x)=\sum_{j}W_j^*xW_j$ it is easy to see that 
$$\eta^\dagger(x)=\sum_{j}\{\sigma^{\frac{1}{2}}W^*_j\sigma^{-\frac{1}{2}} \}^*x \{\sigma^{\frac{1}{2}}W^*_j\sigma^{-\frac{1}{2}}\}.
$$
Let $\Phi$ be a quantum channel and assume that $\sigma$ is an invariant state for $\Phi$, then $\Phi^\dagger$ is again a quantum channel with invariant state $\sigma$:
\[\Phi^\dagger(\mathbf{1})=\sigma^{-\frac{1}{2}} \Phi^*(\sigma^{\frac{1}{2}}\mathbf{1}\sigma^{\frac{1}{2}}) \sigma^{-\frac{1}{2}}=\mathbf{1}, \quad \Phi^{\dagger *}(\sigma)=\sigma^{\frac{1}{2}} \Phi(\sigma^{-\frac{1}{2}} \sigma \sigma^{-\frac{1}{2}})\sigma^{\frac{1}{2}}=\sigma.
\]
The compatibility of KMS-inner product with the convex cone of positive semidefinite matrices allows to decompose every selfadjoint matrix into the difference of two orthogonal positive semidefinite matrices (see \cite[Theorem 3.9]{CM19}): given $x \in \MM$, $x=x^*$ we can write
\[x=\underbrace{\sigma^{-\frac{1}{4}}(\sigma^{\frac{1}{4}}x\sigma^{\frac{1}{4}})_+ \sigma^{-\frac{1}{4}}}_{:=x_{\sigma,+}}-\underbrace{\sigma^{-\frac{1}{4}}(\sigma^{\frac{1}{4}}x\sigma^{\frac{1}{4}})_- \sigma^{-\frac{1}{4}}}_{:=x_{\sigma,-}}
\]
where $(y)_+$ ($(y)_-$) denotes the positive (negative) part of a selfadjoint operator $y \in \MM$ in the sense of functional calculus. It is easy to see that $x_{\sigma,\pm} \geq 0$ and that they are orthogonal with respect to the KMS-product. A simple consequence is the following useful fact which will be used in the proof of our main theorem. With a slight abuse of notation, given a map $\eta$ acting on $\MM$, we denote by $\|\eta\|_2$ the operator norm of $\eta$ induced considering $\MM$ endowed with the norm $\| \cdot \|_2$.
\begin{lemma} \label{lem:KMSnormineq}
Let $\eta_1$ and $\eta_2$ two positive maps defined on $\MM$; then \[\|\eta_1-\eta_2\|_2 \leq \|\eta_1+\eta_2\|_2.
\]
\end{lemma}
\begin{proof}
Every $x \in \MM$ can be decomposed into the sum of two selfadjoint operators: its real part $\Re{(x)}:=(x+x^*)/2$ and its imaginary part $\Im{(x)}=(x-x^*)/(2i)$ and the KMS-norm is compatible with this decomposition, in the sense that
\[\|x\|_2^2=\|\Re{(x)}\|_2^2+ \|\Im{(x)}\|_2^2.
\]
Since $\eta_1$ and $\eta_2$ are positive, $\eta_1-\eta_2$ and $\eta_1+\eta_2$ are real, meaning that they preserve the real subspace of selfadjoint operators. This two facts imply that $\eta_1-\eta_2$ and $\eta_1+\eta_2$ attain their norms on the set of selfadjoint operators of norm one. Let us consider $x \in \MM$, $x=x^*$ such that $\|x\|_2=1$ and $\|(\eta_1-\eta_2)(x)\|_2=\|\eta_1-\eta_2\|_2$; if we define $|x|_\sigma:=x_{\sigma,+}+x_{\sigma,-}$, we have that $\||x|\|_2=1$ and
\[\begin{split}
    &\|\eta_1-\eta_2\|_2^2=\|\eta_1-\eta_2(x)\|_2^2=\|(\eta_1(x_{\sigma,+})+\eta_2(x_{\sigma,-}))-(\eta_1(x_{\sigma,-})+\eta_2(x_{\sigma,+}))\|_2^2\\
    &= \|\eta_1(x_{\sigma,+})+\eta_2(x_{\sigma,-})\|_2^2+\|\eta_1(x_{\sigma,-})+\eta_2(x_{\sigma,+})\|_2^2-2 \underbrace{\langle \eta_1(x_{\sigma,+})+\eta_2(x_{\sigma,-}),\eta_1(x_{\sigma,-})+\eta_2(x_{\sigma,+}) \rangle }_{\geq 0} \\
    & \leq \|\eta_1(x_{\sigma,+})+\eta_2(x_{\sigma,-})\|_2^2+\|\eta_1(x_{\sigma,-})+\eta_2(x_{\sigma,+})\|_2^2+2\langle \eta_1(x_{\sigma,+})+\eta_2(x_{\sigma,-}),\eta_1(x_{\sigma,-})+\eta_2(x_{\sigma,+}) \rangle \\
    &=\|\eta_1+\eta_2(|x|_\sigma)\|_2^2 \leq \|\eta_1+\eta_2\|_2^2.
\end{split}
\]
\end{proof}

\subsection{Output Process of Quantum Markov Chains} \label{sub:qmc}
In the input-output formalism \cite{Gardiner2004} a quantum Markov chain is described as a quantum system interacting sequentially with a chain of identically prepared ancillary systems (the input). After the interaction, the ancillary systems form the output process, and can be measured to produce a stochastic detection record called quantum trajectory
. More formally, let us denote by $\hh=\cc^d$ the system Hilbert space and by $\hh_a$ the ancilla space, and assume that the latter is prepared in the initial state $|\chi \rangle \in \hh_a$; the interaction between the system and a single ancilla is described by a unitary operator $U: \hh \otimes \hh_a \rightarrow \hh \otimes \hh_a$. The reduced evolution of the state of the system after one interaction with the ancillas is described by the following quantum channel called the \emph{transition operator}
\[\Phi^*: \rho \mapsto \tr_{\hh_a}(U \rho \otimes \ket{\chi}\bra{\chi} U^*)
\]
Any orthonormal basis $\{\ket{i}\}_{i \in I}$ for $\hh_a$, induces a Kraus decomposition for $\Phi^*$:
\[\Phi^*(x)=\sum_{i \in I} V_i x V_i^*, \quad V_i:= \bra{i} U \ket{\chi}, \qquad x \in \MM.
\]
If the system is initially prepared in the state $\rho$, after $n$ times steps its state is given by 
$\Phi^{*n}(\rho)$. In general, hypothesis (H) does not guarantee convergence to stationary, i.e. $\lim_{n\to \infty} \Phi^*(\rho) = \sigma$. However, if $\Phi$ primitive (i.e. it satisfies hypothesis (H) and is aperiodic), then any initial state converges to the stationary state. 


Suppose now that, after every interaction between the system and the chain of ancillas, we perform a measurement corresponding to the basis $\{\ket{i}\}_{i \in I}$. The sequence of the outcomes of the measurements $(X_n)_{n \in \nn}$ is a classical $I$-valued stochastic process whose law is uniquely determined by the following collection of finite dimensional distributions: for every $n \in \nn$, $i_1,\dots, i_n \in I$
\begin{equation}\label{eq:law.process}
    \PP_\rho(X_1=i_1,\dots, X_n=i_n)=\tr(V_{i_n} \cdots V_{i_1} \rho V_{i_1}^* \cdots V_{i_n}^*).
\end{equation}
Note that if the system starts in the invariant state $\sigma$, the law of the outcome process is stationary: for every $n,k \in \nn$, $i_1,\dots, i_k \in I$
\[\begin{split}
    &\PP_\sigma(X_{n}=i_1,\dots,X_{n+k-1}=i_k)=\tr(V_{i_k} \dots V_{i_1} \Phi^{*(n-1)}(\sigma)V_{i_1}^* \cdots V_{I_k}^*)\\
    &=\tr(V_{i_k} \dots V_{i_1} \sigma V_{i_1}^* \cdots V_{I_k}^*)=\PP_\sigma(X_1=i_1,\dots, X_{k}=i_k).
\end{split}
\]
We denote by $\pi$ the law of a single measurement under the stationary measure $\PP_\sigma$, i.e. $\pi(i)=\tr(\sigma V_i^*V_i)$ for every $i \in I$.

From \eqref{eq:law.process} it follows that the joint distribution of two measurements at different times $m>n \in \nn$ is given by
\begin{equation}\label{eq:correlations}
    \PP_\rho(X_n=i,X_m=j)=\tr(V_i^*\Phi^{*n-1}(\rho)V_i \Phi^{m-n-1}(V_j^*V_j)).
\end{equation}
If $\Phi$ is primitive then $\Phi^{k}(V_j^*V_j)$ converges to $\pi(j)\mathbf{1}$ for large $k$ so the correlations between $X_n$ and $X_m$ decay exponentially with $m-n$. In fact it has been shown that hypothesis (H) suffices to establish several ergodic results \cite{AGS15,CP15,GH15}: given any function $f: I \rightarrow \rr $, the process $(f(X_n))_{n \in \nn}$ satisfies a strong law of large numbers, a central limit theorem and a large deviation principle. In particular
\begin{equation}
\lim_{n \rightarrow +\infty} \frac{1}{n} \sum_{k=1}^{n} f(X_k)=\pi(f)\quad   \PP_\rho \text{-a.s.}
\end{equation}
where $\pi(f)=\sum_{i \in I} f(i) \pi(i)$. For the reader interested in what happens removing assumption (H), we refer to \cite{CGH21,Gi22}; we will come back to this in Section \ref{sub:redqc}.

Despite the fact that the asymptotic behaviour of the process $(f(X_n))_{n \in \nn}$ is rather well understood, less is known about its finite time properties, with the notable exception of the recent concentration results for continuous-time Markov dynamics \cite{BR21}. The main goal of the present work is to derive alternative concentration bounds, i.e. upper bounds on the probability that $\frac{1}{n} \sum_{k=1}^{n} f(X_k)$ deviates from the limit value more than $\gamma>0$
\begin{equation}
    \PP_\rho \left (\frac{1}{n} \sum_{k=1}^{n} f(X_k)\geq \pi(f)+\gamma \right). 
\end{equation}
Note that by replacing $f$ with $-f$ one obtains an upper bound for the probability of left deviations from the limit value and using the union bound, one can easily control deviations on both sides.


\section{Quantum Markov Chains} \label{sec:QMC}

\subsection{Bernstein-type Inequality}

In this section we prove a Bernstein-type inequality which provides a sub-Gaussian bound for small deviations and a subexponential one for bigger deviations. The inequality involves the spectral gap of the multiplicative symmetrization of the transition operator $\Phi$, and the stationary variance and range of the function $f$. The strategy of the proof is inspired by the classical result \cite[Theorem 3.3]{Le98} and relies on perturbation theory and spectral analysis, which is also the approach used in \cite{CP15,GH15} for proving the law of large numbers and the large deviation principle for the process $(f(X_k))_{k \in \nn}$.
\begin{theo} \label{theo:main1}
Assume that $\Psi:=\Phi^\dagger\Phi$ is irreducible, and let $f:I \rightarrow \rr$ such that $\pi(f)=0$, $\pi(f^2)=b^2 $ and $\| f \|_\infty:=\max_{i \in I} |f(i)|=c$ for some $b,c>0$. Then for every $\gamma > 0$, $n \geq 1$
\begin{equation} \label{eq:quantbound}
\PP_\rho \left (\frac{1}{n} \sum_{k=1}^{n} f(X_k)\geq\gamma \right) \leq N_\rho \exp \left ( -n \frac{\gamma^2 \varepsilon}{6b^2}h\left ( \frac{10c \gamma}{3b^2}\right )\right )
\end{equation}
where $\varepsilon$ is the spectral gap of $\Psi$, $N_\rho:=\|\sigma^{-\frac{1}{2}}\rho \sigma^{-\frac{1}{2}}\|_2$ and $h(x)=(\sqrt{1+x}+x/2+1)^{-1}$.
\end{theo}

We remark that a sufficient (but not necessary) condition for $\Psi$ to be irreducible is that $\Phi$ is positivity improving, i.e. that $\Phi(x)>0$ for every $x \geq 0$.

\begin{proof}
We split the proof in 4 steps. 
The strategy is to use Markov inequality to bound the deviation probabilities in terms of the moment generating function (Chernoff bound), which is then bounded using perturbation theory and spectral properties.

\vspace{2mm}

\noindent
\textbf{1. Upper bound for the moment generating function (Laplace transform).} 

\noindent
An easy computation shows that the Laplace transform of $n\bar{f}_n:= \sum_{k=1}^{n} f(X_k)$ can be expressed in terms of the deformed transition operator $\Phi_u(x)=\sum_{i \in I} e^{u f(i)}V_i^* x V_i$, $u>0$:
\begin{equation}
\begin{split}
    \mathbb{E}_\rho[e^{n u\bar{f}_n}]&=\tr(\rho \Phi_u^n(\mathbf{1}))=\tr\left(\sigma^{\frac{1}{2}}(\sigma^{-\frac{1}{2}}\rho \sigma^{-\frac{1}{2}})\sigma^{\frac{1}{2}} \Phi_u^n(\mathbf{1})\right)=\langle \sigma^{-\frac{1}{2}}\rho \sigma^{-\frac{1}{2}}, \Phi_u^n(\mathbf{1}) \rangle \\
    &\leq \|\sigma^{-\frac{1}{2}}\rho \sigma^{-\frac{1}{2}}\|_2 \cdot \|\mathbf{1}\|_2 \cdot \|\Phi_u^n\|_2 =
    \|\sigma^{-\frac{1}{2}}\rho \sigma^{-\frac{1}{2}}\|_2 \cdot \|\Phi_u^n\|_2 .
    \end{split}
\end{equation}
By $\|\Phi_u^n\|_2$ we mean the operator norm of $\Phi_u^n$ induced by $\MM$ endowed with the KMS-norm. The rest of the proof aims to upper bound $\|\Phi_u^n\|_2$; a first remark is that
\[\|\Phi_u^n\|_2\leq \|\Phi_u\|^n_2=\|\Psi_u\|^{\frac{n}{2}}_2=r(u)^{\frac{n}{2}}
\]
where $\Psi_u:=\Phi^\dagger_u \Phi_u$ and $r(u)$ is the spectral radius of $\Psi_u$. Notice that for every $u \in \rr$, $\Phi_u=\sum_{i \in I} (e^{\frac{u f(i)}{2}}V_i)^* x (e^{\frac{u f(i)}{2}}V_i)$ is completely positive, therefore $\Phi^\dagger_u$ and $\Psi_u$ are completely positive too; Proposition \ref{prop:PF} ensures that $r(u)$ is an eigenvalue of $\Psi_u$. Moreover, we can write
\[\Psi_u(x)=\Phi^\dagger_u  \Phi_u(x)=\sum_{i,j \in I} \{e^{\frac{u(f(i)+f(j))}{2}}K_{i,j}\}^* x \{e^{\frac{u(f(i)+f(j))}{2}} K_{i,j}\}
\]
where $K_{i,j}=V_i \sigma^{\frac{1}{2}}V^*_j\sigma^{-\frac{1}{2}} $ are the Kraus operators of $\Psi$. Since $\Psi$ is irreducible by assumption, its Kraus operators satisfy condition 2. in Proposition \ref{prop:PF}.   Since the Kraus operators of $\Psi_u$ are multiples of $K_{i,j}$, they also satisfy condition 2 and therefore  $r(u)$ is an algebraically simple eigenvalue of $\Psi_u$ for every $u \in \rr$.

\bigskip 
\noindent\textbf{2. Perturbation theory.}

\noindent
Direct computations show that $\Psi_u$ is an analytic perturbation of $\Psi=\Psi_0$: first notice that
\begin{equation}
  \Phi_u(x)=\sum_{k=0}^{+\infty}  \frac{u^k}{k!}\underbrace{\sum_{i \in I}f(i)^k V_i^* x V_i }_{=:\Phi^{(k)}}
\end{equation}
and that $\Phi^{(0)}=\Phi$. Therefore we can write
\begin{equation}
    \Psi_u=\Phi_u^\dagger \Phi_u=\left ( \sum_{k=0}^{+\infty} \frac{u^k}{k!}\Phi^{(k)\dagger}\right) \left (\sum_{k=0}^{+\infty} \frac{u^k}{k!}\Phi^{(k)}\right)=\sum_{k=0}^{+\infty} \frac{u^k}{k!} \underbrace{\sum_{l=0}^k \binom{k}{l} \Phi^{(l)\dagger} \Phi^{(k-l)}}_{=:\Psi^{(k)}}
 \end{equation}
and it is easy to see that $\Psi^{(0)}=\Psi$.
If the norms of $\Psi^{(k)}$ are upper bounded by a geometric sequence, that is
\[\| \Psi^{(k)} \|_2 \leq \alpha \beta^{k-1} \text{ for $k \geq 1$}
\]
for some $\alpha,\beta >0$, then perturbation theory (\cite[Section II.2.2]{Ka76}) tells us that for $|u| < (2\alpha \varepsilon^{-1}+\beta)^{-1}$ (where $\varepsilon$ is the spectral gap of $\Psi$), we can expand $r(u)$ around $0$:
\[r(u)=1+\sum_{k=1}^{+\infty}r^{(k)}u^k.\]
Moreover, there is an explicit expression for the coefficients $r^{(k)}$'s:
\[r^{(k)}=\sum_{p=1}^k \frac{(-1)^p}{p} \sum_{\substack{\nu_1+\dots +\nu_p=k,\, \nu_i\geq 1\\
\mu_1+\dots +\mu_p=p-1, \, \mu_j\geq 0}} \frac{1}{\nu_1! \cdots \nu_p!}\TR(\Psi^{(\nu_1)} S^{(\mu_1)}\cdots \Psi^{(\nu_p)}S^{(\mu_p)}),
\]
where $S^{(0)}=-\ket{\mathbf{1}}\bra{\mathbf{1}}$ and for $\mu\geq 1$, $S^{(\mu)}$ is the $\mu$-th power of
\[S^{(1)}=(\Psi-{\rm Id}+\ket{\mathbf{1}}\bra{\mathbf{1}})^{-1}-\ket{\mathbf{1}}\bra{\mathbf{1}}.\]
Notice that $\|S^{(\mu)}\|_2=\varepsilon^{-\mu}$ for $\mu\geq 1$. More details can be found in \cite[Section 2]{Le98}.

\bigskip 
\noindent\textbf{3. Upper bound for $r^{(k)}$s.}

\noindent
The coefficient $r^{(1)}$ can be easily shown to be equal to zero due to the fact that $f$ is centered:
\[|\TR(\Psi^{(1)} S^{(0)})|\leq 2|\langle \mathbf{1},\Phi^{(1)} (\mathbf{1})  \rangle|=2\pi(f)=0.
\]
For bigger values of $k$, we will make extensive use of the estimates in the following lemma; we recall that $b^2=\pi(f^2)$ and $c=\|f\|_\infty$.
\begin{lemma} \label{lem:normub}
Let $k$ and $m$ be two natural numbers; the following bounds hold true:
\begin{enumerate}
    \item $\|\Phi^{(m)\dagger} \Phi^{(k)}\|_2 \leq c^{m+k}$,
    \item $\|\Psi^{(m)}\|_2 \leq {(2c)}^m$,
    \item $\|\Phi^{(m)\dagger}\Phi^{(k)}(\mathbf{1})\|_2 \leq c^{m+k-1}b$ for $m+k \geq 1$,
    \item $\|\Psi^{(m)} (\mathbf{1}) \|_2 \leq 2{(2 c)}^{m-1}b$ for $m \geq 1$,
    \item $|\langle \mathbf{1}, \Psi^{(m)} (\mathbf{1})\rangle| \leq 4{(2c)}^{m-2}b^2$ for $m \geq 2$.
\end{enumerate}
\end{lemma}
\begin{proof}
1. Lemma 1.3 in \cite{OP04} ensures that if we consider a completely positive map $\eta:\MM \rightarrow \MM$ such that $\eta^*(\sigma)\leq \sigma$, then $\|\eta\|_2 \leq 1$. If we apply it to $\Phi$ and $\Phi^\dagger$, we get $\|\Phi^\dagger \Phi \|_2 \leq 1$, which proves equation in point 1. in the case $k=m=0$. If $k+m\geq 1$, notice that we can write
\[\begin{split}
\frac{\Phi^{(m)\dagger} \Phi^{(k)}(x)}{c^{m+k}}&=\underbrace{\sum_{\substack{j,i \in I:\\f(j)^m f(i)^k\geq 0}}\frac{f(j)^m f(i)^k}{c^{m+k}} K_{i,j}^* x K_{i,j}}_{\eta_+(x)}\\
&-\underbrace{\sum_{\substack{j,i \in I:\\f(j)^m f(i)^k< 0}}-\frac{f(j)^m f(i)^k}{c^{m+k}} K_{i,j}^*x K_{i,j}}_{\eta_-(x)}.\\
\end{split}
\]
$\eta_{\pm}$ are completely positive and, since $\|f\|_\infty = c$, they also satisfy $(\eta_++\eta_-)^*(\sigma) \leq \Phi^{*} \Phi^{ \dagger  *} (\sigma) = \sigma$, hence $\|\eta_++\eta_- \|_2 \leq 1$ and we get the thesis using Lemma \ref{lem:KMSnormineq}.

\vspace{2mm}
\noindent 
2. By the explicit form of $\Psi^{(m)}$ and point 1. we get
\[\|\Psi^{(m)}\|_2 \leq \sum_{l=0}^m \binom{m}{l} \|\Phi^{(m)\dagger} \Phi^{(m-l)}\|_2 \leq \left (\sum_{l=0}^m \binom{m}{l} \right )  c^{m}={(2c)}^m.
\]

\vspace{2mm}
\noindent 
3. Let us introduce the following selfadjoint operator

\[F^{(m,k)}:=\Phi^{(m)\dagger} \Phi^{(k)}(\mathbf{1})=\sum_{i,j \in I}f(i)^k f(j)^m K_{i,j}^*K_{i,j},\]

where we recall that $K_{i,j}=V_i \{\sigma^{\frac{1}{2}}V^*_j\sigma^{-\frac{1}{2}} \}$ are the Kraus operators of $\Psi$; notice that $F^{(m,k)}$ is a convex combination with operator weights $K_{i,j}^* K_{i,j}$ of the matrices $f(i)^k f(j)^m \mathbf{1}$. We get that
\[\begin{split}
\|\Phi^{(m)\dagger} \Phi^{(k)}( \mathbf{1}) \|_2&=\tr(\sigma^{\frac{1}{2}}F^{(m,k)}\sigma^{\frac{1}{2}}F^{(m,k)})^{\frac{1}{2}}\leq \tr \left (\sigma \left (F^{(m,k)}\right)^2 \right )^{\frac{1}{2}} \\
&\leq \left (\sum_{i,j \in I} f(i)^{2k}f(j)^{2m}\tr(\sigma K_{i,j}^*K_{i,j}) \right )^{\frac{1}{2}}\\
\end{split}
\]
where we used Cauchy-Schwarz inequality for the trace and operator Jensen's inequality (see for instance \cite[Theorem 2.1]{HP03}), since $t^2$ is operator convex on the whole real line. Notice that the last term of the previous equation can be expressed as $\tilde{\pi}(f(i)^{2k}f(j)^{2m})^{\frac{1}{2}}$, where $\tilde{\pi}$ is the probability measure on $I^2$ defined as $\tilde{\pi}(i,j)=\tr(\sigma K_{i,j}^*K_{i,j})$; an easy computation shows that $\tilde{\pi}$ has marginals equal to $\pi$. If $k=0$ (and analogously if $m=0$) we get that
\[\tilde{\pi}(f(i)^{2k}f(j)^{2m})^{\frac{1}{2}}=\pi(f^{2m})^{\frac{1}{2}} \leq c^{m-1}b,
\]
otherwise, if both $m$ and $k$ are bigger or equal than $1$, we can still get the same bound:
\[\tilde{\pi}(f(i)^{2m}f(j)^{2k})^{\frac{1}{2}} \leq \pi(f^{2m})^{\frac{1}{2}}c^{k} \leq c^{m+k-1}b.
\]
4. Because of point 3., we obtain
\[
    \|\Psi^{(m)} (\mathbf{1}) \|_2 \leq \sum_{l=0}^m \binom{m}{l} \|\Phi^{(l)\dagger} \Phi^{(m-l)} (\mathbf{1})\|_2 \leq 2 {(2c)}^{m-1}b.
\]\\
5. From the estimate in point 3. and the explicit expression of $\Psi^{(m)}$ we have that:
\[
\begin{split}
    |\langle \mathbf{1}, \Psi^{(m)}(\mathbf{1}) \rangle|&\leq \sum_{l=0}^m \binom{l}{m} |\langle \Phi^{(l)}(\mathbf{1}),\Phi^{(m-l)}(\mathbf{1}) \rangle| \\
    & \leq \sum_{l=0}^m \binom{l}{m} \|\Phi^{(l)}(\mathbf{1})\|_2 \|\Phi^{(m-l)}(\mathbf{1}) \|_2 \leq 4 {(2c)}^{m-2}b^2.
\end{split}
\]
In case $l=0$ or $l=m$ we do not make use of estimate in point 3., but the upper bound follows from the observation that $\langle\mathbf{1},\Phi^{(m)}(\mathbf{1}) \rangle=\langle\Phi^{(m)}(\mathbf{1}),\mathbf{1} \rangle=\pi(f^m)$.
\end{proof}

Fix $k\geq 2$. Since $\mu_1+\dots +\mu_p=p-1$, there must be a certain $j \in \{1,\dots,p\}$ such that $\mu_j=0$; by the cyclicity of the trace, without loss of generality we can assume that $\mu_p=0$. Using the estimates in Lemma \ref{lem:normub} we get that for $p=1$
\[|\TR(\Psi^{(k)} S^{(0)})|= |\langle \mathbf{1},\Psi^{(k)} (\mathbf{1}) \rangle|\leq 4{(2c)}^{k-2} b^2 .
\]
Hence
\begin{equation} \label{eq:est1}
\frac{|\TR(\Psi^{(k)} S^{(0)})|}{k!}\leq 2 {(2c)}^{k-2}b^2  \leq \frac{b^2}{c}\left ( \frac{2c}{\varepsilon} \right )^{k-1}. 
\end{equation}
If $p\geq 2$, then
\[\begin{split}
    &|\TR(\Psi^{(\nu_1)} S^{(\mu_1)}\cdots \Psi^{(\nu_p)}S^{(\mu_p)})|=|\langle \mathbf{1},\Psi^{(\nu_1)} S^{(\mu_1)}\cdots \Psi^{(\nu_p)} (\mathbf{1}) \rangle| \\
    &\leq \varepsilon^{1-p} {(2c)}^{k-\nu_1-\nu_p} \|\Psi^{(\nu_1)}(\mathbf{1})\|_2\|\Psi^{(\nu_p)}(\mathbf{1})\|_2 \leq \varepsilon^{1-p} {(2c)}^{k-\nu_1-\nu_p}  2^2{(2c})^{\nu_1+\nu_p-2}b^2  \\
    &\leq \varepsilon^{1-p} 2^{k} c^{k-2}b^2.
\end{split}
\]
Hence
\begin{equation} \label{eq:est2}
   \frac{|\TR(\Psi^{(\nu_1)} S^{(\mu_1)}\cdots \Psi^{(\nu_p)}S^{(\mu_p)})|}{\nu_1!\cdots \nu_p!} \leq \frac{2b^2}{c}\left ( \frac{2c}{\varepsilon} \right )^{k-1}.
\end{equation}

For $k \geq 3$, the following upper bound holds true (we refer to \cite{Le98} for more details): 
\[\sum_{p=1}^k \frac{1}{p} \sum_{\substack{\nu_1+\dots +\nu_p=k,\, \nu_i\geq 1\\
\mu_1+\dots +\mu_p=p-1, \, \mu_j\geq 0}} \leq 5^{k-2}.
\]
We conclude that for $k \geq 3$
\begin{equation} \label{eq:estrk}
|r^{(k)}| \leq \frac{2b^2}{5c}\left ( \frac{10c}{\varepsilon} \right )^{k-1}.
\end{equation}
For $k=2$ we upper bound $|r^{(2)}|$ in the following way: indeed, thank to equations \eqref{eq:est1} and \eqref{eq:est2}, we obtain
\begin{equation} \label{eq:dynvariance}
|r^{(2)}|=\left | \frac{\TR(\Psi^{(2)} S^{(0)})}{2}-\langle \Psi^{(1)}(\mathbf{1}), S^{(1)}\Psi^{(1)}(\mathbf{1})\rangle \right |\leq  2b^2+\frac{4b^2}{\varepsilon} \leq \frac{6 b^2}{\varepsilon}.
\end{equation}
Putting everything together we get that for $0 \leq u < \varepsilon/(10c)$
\[\begin{split}
    r(u) &\leq 1+ \sum_{k\geq 1} |r^{(k)}|u^k \leq 1+ u^2\frac{6b^2}{\varepsilon}\sum_{k\geq 0} \left ( \frac{10cu}{\varepsilon} \right )^{k}\\
    &\leq 1+ u^2\frac{6b^2}{\varepsilon} \left ( 1-\frac{10cu}{\varepsilon} \right )^{-1} \leq e^{u^2\frac{6b^2}{\varepsilon} \left ( 1-\frac{10cu}{\varepsilon} \right )^{-1}}.\\
\end{split}
\]

\bigskip 

\noindent 
\textbf{4. Chernoff bound and Fenchel-Legendre transform.}

\noindent
Now that we have an upper bound for the Laplace transform, we apply the usual machinery of the Chernoff bound: using Markov inequality we obtain that for every $0 \leq u < \varepsilon/(10c)$
\[\begin{split}
\PP_\rho(\bar{f}_n\geq\gamma)& =\PP_\rho(e^{n u\bar{f}_n}\geq e^{n u\gamma}) \leq e^{-nu\gamma} \mathbb{E}_\rho[e^{nu\bar{f}}]\\
&\leq N_\rho \exp \left \{-n\left (\gamma u-u^2\frac{3b^2}{\varepsilon} \left ( 1-\frac{10cu}{\varepsilon} \right )^{-1} \right ) \right \}.\\
\end{split}
\]
Taking the infimum of the rhs over admissible values of $u$, we obtain
\[\PP_\rho\left (\frac{1}{n} \sum_{k=1}^{n} f(X_k) \geq \gamma \right ) \leq N_\rho \exp \left ( -n \frac{\gamma^2 \varepsilon}{6b^2}h\left ( \frac{10c \gamma}{3b^2}\right )\right )
\]
for $h(x)=(\sqrt{1+x}+x/2+1)^{-1}$.
\end{proof}

\subsubsection{Comparison to the classical concentration bound}
For easier comparison between the bound in Theorem \ref{theo:main1} and the classical Markov chains results in \cite{Le98}, we report the latter below. Let $X_n$ an irreducible Markov chain on the (finite) state space $E$ with transition matrix $P$, initial law $\nu$ and invariant measure $\pi$ and let $f:E \rightarrow \mathbb{R}$ be a bounded function with $\pi(f)=0$, $\pi(f^2)=b^2$ and $\|f\|_\infty=c$. Then 
\begin{enumerate}
    \item if $P$ is selfadjoint, then
    \[\PP_\nu\left (\frac{1}{n}\sum_{k=1}^{n}f(X_k)\geq \gamma \right ) \leq \left \| \frac{d\nu}{d\pi}\right \|_2 e^{c \varepsilon/5} \exp \left ( -n \frac{\gamma^2 \varepsilon}{2b^2}h\left ( \frac{5c \gamma}{b^2}\right )\right ),
    \]
    where $\varepsilon$ is the spectral gap of $P$,
    \item if $P^\dagger  P$ is irreducible, then 
    \begin{equation} \label{eq:classbound}
    \PP_\nu\left (\frac{1}{n}\sum_{k=1}^{n}f(X_k)\geq \gamma \right ) \leq \left \| \frac{d\nu}{d\pi}\right \|_2  \exp \left ( -n \frac{\gamma^2 \varepsilon}{4b^2}h\left ( \frac{5c \gamma}{b^2}\right )\right ),
    \end{equation}
    where $\varepsilon$ is the spectral gap of $P^\dagger  P$.
\end{enumerate}
The difference in the constants appearing in the bound in equation \eqref{eq:quantbound} and in the classical one (equation \eqref{eq:classbound}) comes from the worse upper bound one can get for $r^{(2)}$ (equation \eqref{eq:dynvariance}) in this more general setting.

In order to obtain the result for $P$ selfadjoint, a crucial observation is that, in this case, $P_u:=PE_u$ (where $E_u=(\delta_{xy}e^{uf(x)})_{x,y \in E}$), is similar to a selfadjoint matrix: indeed, $PE_u=E_u^{-1/2}(E_u^{1/2}P E_u^{-1/2}) E_u^{-1/2}$. However this is not the case for $\Phi_u$, as the following elementary example shows. Let us consider a three dimensional quantum system, an orthonormal basis $\{\ket{k}\}\}_{k=0}^2$ and the quantum channel $\Phi$ with the following Kraus operators:
\[V_{k,k+1}=\frac{1}{2} \ket{k+1}\bra{k}, \quad V_{k,k-1}= \frac{1}{2} \ket{k-1}\bra{k}, \quad \text{ for } k=0,1,2,
\]
where $k+1$ and $k-1$ are understood as modulo $3$. In this case the index set of Kraus operators is $I=\{(k,k+1),(k,k-1): k=0,1,2\}$. $\Phi$ is selfadjoint, but if we pick the function $f(k,k+1)=a$ and $f(k,k-1)=b$ for some real numbers $a \neq b$, then then corresponding perturbation $\Phi_u$ has complex eigenvalues for every $u >0$, hence it cannot be similar to a selfadjoint map. A way of better understanding the difference between the classical and the quantum setting highlighted by the example is to notice that $\Phi$ is a quantum dilation of the symmetric random walk on a three vertices ring: indeed, $\Phi$ preserves the algebra
\[\Delta:=\left \{x \in \MM: x=\sum_{k=0}^2 x(k) \ket{k}\bra{k} \right \},\]
which is isomorphic to the algebra of functions on three points \[\ell^\infty(E):=\{g:E \rightarrow \cc\}, \quad E=\{0,1,2\},\]
and its restriction to it is given by the the transition matrix corresponding to the symmetric random walk on three vertices, i.e. $P=(p_{lk})$ where $p_{lk}=\frac{1}{2} (\delta_{k,l+1} + \delta_{k,l-1})$; $f$ is a function of the jumps of the random walk and not of its states and the restriction of $\Phi_u$ to the diagonal algebra $\Delta$ is given by a perturbation of $P$ of the form $P_u=(p_{lk}e^{u f(l,k)})$, which belongs to a more general class of perturbations of $P$ than the one considered in \cite{Le98}. In Subsection \ref{ss:cMC} we will show how this fact allows to prove new concentration inequalities for fluxes of classical Markov chains. 
\subsection{Hoeffding-type Inequality}
In this section we prove a second quantum concentration bound inspired by classical result \cite{GO02} which relies on the application of a fundamental inequality for centered bounded random variables (Hoeffding's inequality) and the fact that $\frac{1}{n}\sum_{i=0}^{n-1}f(X_i)$ can be decomposed into a martingale with bounded increments and a bounded reminder. The same martingale decomposition was also used in \cite{AGS15} for proving the law of large number and the central limit theorem for $(f(X_k))_{k \in \nn}$. The following inequality does not involve any measure of the variance of the function $f$ at stationarity and, instead of the spectral gap of the multiplicative symmetrization of the quantum channel, it depends on the norm of the pseudoresolvent $({\rm Id}-\Phi)^{-1}$. We remark that, contrary to Theorem \ref{theo:main1}, in this case there are no further assumptions on $\Phi$.
\begin{theo} \label{theo:main2}
For every $f:I \rightarrow \mathbb{R}$ such that $\pi(f)=0$ and $\|f\|_\infty=c$ for some $c >0$, then for every $\gamma > 0$
\begin{equation}
\PP_\rho \left (\frac{1}{n} \sum_{k=1}^{n} f(X_k)\geq \gamma \right) \leq \exp \left ( - \frac{(n \gamma -2G)^2}{2(n-1)G^2}\right ) \text{ for } n \gamma \geq 2G,
\end{equation}
where $G=(1+\|({\rm Id}-\Phi)^{-1}_{|\mathcal{F}}\|_\infty)c$ and $\mathcal{F}:=\{x \in \MM: \tr(\sigma x)=0\}$.
\end{theo}
With a slight abuse of notation, we write $\|({\rm Id}-\Phi)^{-1}_{|\mathcal{F}}\|_\infty$ for denoting the operator norm of $({\rm Id}-\Phi)^{-1}_{|\mathcal{F}}$ induced by considering the uniform norm on ${\cal F}$. Notice that for $n=1$, the constraint on $\gamma$ implies that $\gamma \geq 2c=2\|f\|_\infty$ and consequently $\PP_\rho ( f(X_1)\geq \gamma) =0$.

\begin{proof} For clarity we split the proof in three steps.

\vspace{2mm}

\noindent
\textbf{1. Poisson equation.}\\
We start by discussing the quantum trajectories Markov process and its associated Poisson equation, which is a key tool in the proof. The pair $(X_n, \rho_n)$ consisting of the $n$-th measurement outcome $X_n\in I$ and the conditional system state $\rho_n$ is a 
Markov chain with 
\[\PP_\rho \left (X_{n+1}=i, \rho_{n+1}=\left. \frac{V_i \omega V_i^*}{\tr(V_i \omega V_i^*)} \right |X_n=j, \rho_n =\omega\right )= \tr(V_i \omega V_i^*).
\]
and initial condition 
$$
X_1 =i \quad {\rm and} \quad 
\rho_1 =  \frac{V_i \rho V_i^*}{\tr(V_i \rho V_i^*)} \text{ with probability } \tr(V_i \rho V_i^*).
$$

Its transition operator $P$ is given by 
\[ Pg(i,\omega):=\sum_{j \in I} g\left (j, \frac{V_j \omega V_j^*}{\tr(V_j \omega V_j^*)} \right )\tr(V_j \omega V_j^*).
\]

The associated Poisson equation is 
\begin{equation}
F(i,\omega)=g(i,\omega)-Pg(i,\omega)
\end{equation}
where $F(i,\omega)$ is a given function, and one is interested in finding $g(i,\omega)$. We will provide an heuristical explanation on how to find a solution of the Poisson equation. Whenever it is well defined, a natural candidate for $g$ is the function $g(i,\omega)=\sum_{n \geq 1} \mathbb{E}[F(X_n,\rho_n)|X_1=i,\rho_1=\omega]=\sum_{n \geq 0}P^{n}(F)(i,\omega)$; indeed
\[g(i,\omega)-Pg(i,\omega)=\sum_{n \geq 0} P^n F(i,\omega)-\sum_{n \geq 1} P^n F(i,\omega)=F(i,\omega).
\]
We now consider a function $F$ which does not depend on the second argument $F(i,\omega) =f (i)$. Using the explicit expression of $P$ we can write
\[g(i,\omega)=f(i)+\tr\left ( \sum_{n \geq 0} \Phi^n({\bf F}) \omega \right )= f(i)+\tr( ({\rm Id} -\Phi)^{-1}({\bf F}) \omega)
\]
where ${\bf F}=\sum_{i \in I} f(i) V_i^*V_i$. With this in mind, the following steps should look very reasonable. We recall that $c:=\|f\|_\infty$.
\begin{lemma} \label{lemm:poiss}
The equation
\begin{equation} \label{eq:poiss}
({\rm Id} -\Phi)(A)={\bf F}  
\end{equation}
admits a solution and all the solutions differ for a multiple of the identity. We denote by $A_f$ the unique solution such that $\tr(\sigma A_f)=0$; we have that $\|A_f\|\leq c \|({\rm Id}-\Phi)^{-1}_{|\mathcal{F}}\|_\infty$ where $\mathcal{F}:=\{x \in \MM: \tr(\sigma x)=0\}$.
\end{lemma}
\begin{proof}
Since $\Phi$ is irreducible, ${\rm dim}(\ker({\rm Id}-\Phi))=1$ and $1$ is also the codimension of ${\rm rank}({\rm Id}-\Phi)$; it is easy to see that ${\rm rank}({\rm Id}-\Phi)=\{x \in \MM: \tr(\sigma x)=0\}:={\cal F}$ and equation \eqref{eq:poiss} admits a solution because $\tr(\sigma {\bf F})=\pi(f)=0$. Two solutions of \eqref{eq:poiss} differ for an element of $\ker({\rm Id}-\Phi)=\mathbb{C} \mathbf{1}$. Finally, ${\rm Id}-\Phi$ is a bijection on ${\cal F}$, hence it makes sense to write $A_f = ({\rm Id}-\Phi)^{-1}({\bf F})$ and we get that $\|A_f\| \leq \|({\rm Id}-\Phi)^{-1}_{|\mathcal{F}}\|_\infty \|{\bf F}\|$. We only need to show that $\|{\bf F}
\| \leq c$: notice that ${\bf F}$ can be written as the difference of two positive semidefinite matrices as
\[{\bf F}=\sum_{\substack{i \in I \\ f(i)\geq 0}} f(i) V_i^* V_i - \sum_{\substack{i \in I \\ f(i)< 0}} |f(i)| V_i^* V_i,\]

hence by \cite[Corollary 3.17]{Zh02} we have that $\|{\bf F}\| \leq \|\sum_{ i \in I} |f(i)| V_i^*V_i\|\leq c$.
\end{proof}
It is now an easy computation to verify that if we take $g(i, \omega):=f(i)+ \tr(A_f \omega)$, then 
\[f(i)=g(i, \omega)-Pg(i,\omega).
\]
Moreover $\|g\|_\infty \leq c(1+\|({\rm Id}-\Phi)^{-1}_{|\mathcal{F}}\|_\infty)$.

\bigskip 
\noindent\textbf{2. Hoeffding's inequality.}\\
Thank to the previous step, we can write $\sum_{k=1}^n f(X_k)$ as a martingale with bounded increments and a bounded reminder:
\[\begin{split}
\sum_{k=1}^n f(X_k)&=\sum_{k=1}^n g(X_k,\rho_k)-\mathbb{E}_\rho[g(X_{k+1}, \rho_{k+1})|X_k,\rho_k]\\
&=\sum_{k=2}^{n} \underbrace{g(X_k,\rho_k)-\mathbb{E}_\rho[g(X_{k}, \rho_{k})|X_{k-1},\rho_{k-1}]}_{D_k}\\
&+g(X_1, \rho_1)-\mathbb{E}_\rho[g(X_{n+1}, \rho_{n+1})|X_n,\rho_n].\\
\end{split}
\]
We can easily bound from above the Laplace transform of $\sum_{k=1}^n f(X_k)$ in the following way: for every $u>0$
\[\begin{split}
        \mathbb{E}_\rho[e^{u\sum_{k=1}^n f(X_k)}]& \leq e^{2\|g\|_\infty u} \mathbb{E}_\rho[e^{u\sum_{k=2}^{n-1} D_k} \mathbb{E}_\rho[e^{u D_n}|X_{n-1}, \rho_{n-1},\dots, X_1,\rho_1]]\\
    &\leq e^{2\|g\|_\infty u} e^{\|g\|_\infty^2 u^2/2} \mathbb{E}_\rho[e^{u\sum_{k=2}^{n-1} D_k}]
\end{split}
\]
where in the last equation we used Hoeffding's Lemma \cite[Lemma 2.2]{BLM13}. By induction, we get
\[\mathbb{E}_\rho[e^{u\sum_{k=1}^n f(X_k)}] \leq exp(2\|g\|_\infty u+(n-1)\|g\|_\infty^2 u^2/2).
\]

\bigskip 
\noindent\textbf{3. Fenchel-Legendre transform}\\
As in point 4. of the proof of Theorem \ref{theo:main1}, the final statement follows using Markov inequality and optimizing (remembering that $u>0$). 
\end{proof}

\section{Quantum Counting Processes} \label{sec:QCP}
In this section we consider continuous-time concentration bounds for counting measurements in the output of a quantum Markov process. In this context, alternative concentration bounds  have been recently obtained in \cite{BR21} using functional inequalities. While some concentration bounds involving easily computable quantities have been proved in the diffusive case, for quantum counting processes they are still missing. In this section we show that the same perturbative analysis as in the proof of Theorem \ref{theo:main1} can be used to take a first step towards filling this gap. More precisely, we consider a continuous-time quantum Markov process with GKLS generator \cite{Gardiner2004} given by
\begin{equation}\label{eq:GKLS}
{\cal L}(x)+i[H,x]=\sum_{i \in I} L^*_i x L_i -\frac{1}{2}\{L^*_iL_i,x\}, \quad x \in \MM
\end{equation}
where $[x,y]:=xy-yx$, $\{x,y\}:=xy+yx$, $H \in \MM$ is selfadjoint and $\{L_i\}_{i \in I}$ is a finite collection of operators in $\MM$. It is well known that the family of maps $\Phi_t:=e^{t {\cal L}}$ for $t\geq 0$ is a uniformly continuous semigroup of quantum channels. Analogously to the discrete time case, we need to make an irreducibility assumption.
\begin{equation}\tag{H$^\prime$}
    \text{There exists a unique faithful state $\sigma$ such that ${\cal L}^*(\sigma)=0$.}
\end{equation}
Hypothesis (H$^\prime$) means that $(\Phi_t)_{t \geq 0}$ is irreducible. There are many ways of defining quantum counting processes: we will follow the formulation of Davis and Srinivas \cite{Da76,DS81} and we refer to \cite{BR21} and references therein for their definition thorough quantum stochastic calculus and quantum filtering or how to characterize them using stochastic Schr{\"o}dinger equations. Before doing that, it is convenient to introduce some notation: we define the completely positive map ${\cal J}_i(x)=L_i^*x L_i$ and the semigroup of completely positive maps $e^{t {\cal L}_0}(x)=e^{t G^*} x e^{t G}$, where $G:= i H-\frac{1}{2} \sum_{i \in I} L_i^*L_i$. As the process that we studied in previous sections, also the one we are about to define can be used to model the stochastic process coming from indirect measurements performed on a certain quantum system. We assume that the system is coupled to $|I|$ detectors: when the $i$-th detector clicks, the state of the system evolves according to the map $\rho \mapsto {\cal J}^*_i(\rho)/\tr({\cal J}^*_i(\rho))$, while in between detections the evolution is dictated by $e^{t {\cal L}^*_0}(\rho)$. At time $t$ the instantaneous intensity corresponding to the $i$-th detector is given by $\tr({\cal J}_i^*(\rho_{t-}))$, where $\rho_t$ is the stochastic process describing the evolution of the state of the system. More formally, we can use Dyson's expansion of the semigroup $\Phi_t$ in order to define a proper probability measure on $\Omega_t=\{(t_1,i_1,\dots,t_k,i_k): k \in \nn, 0 \leq t_1 \leq \dots \leq t_k \leq t, i_1,\dots, i_k \in I\}$:
\[\begin{split}
\Phi_t(\rho)&=e^{t{\cal L}^*_0}(\rho)\\
&+\sum_{k=1}^{+\infty} \sum_{i_1,\dots, i_k \in I} \int_0^t \cdots \int_0^{t_2} e^{(t-t_k){\cal L}^*_0}{\cal J}^*_{i_k} \cdots {\cal J}^*_{i_k} e^{t_1{\cal L}^*_0}(\rho) dt_1 \cdots dt_k.
\end{split}
\]
Notice that
\[\Omega_t=\{\emptyset\} \cup \bigcup_{k=1}^{+\infty} I^k \times \{(t_1,\dots,t_k)\in [0,t]^k: t_1 \leq \dots \leq t_k\},
\]
so there is a natural way of endowing it with the $\sigma$-field induce by considering the $\sigma$-field of all the subsets on $\{\emptyset\}$ and $I$, and the Lebesgue $\sigma$-field on $\{(t_1,\dots,t_k)\in [0,t]^k: t_1 \leq \dots \leq t_k\}$. We denote by $d\mu$ the unique measure such that $\mu(\{\emptyset\})=1$ and $\mu(\{(i_1,\dots, i_k)\times B)$ is the Lebesgue measure of $B$ for every $(i_1,\dots, i_k) \in I^k$, $B \subseteq \{(t_1,\dots,t_k)\in [0,t]^k: t_1 \leq \dots \leq t_k\}$. Notice that the following normalization condition holds true
\[\begin{split}
    1&=\tr(\Phi_t(\rho))=\\
    &=\tr(e^{t{\cal L}^*_0}(\rho))+\sum_{k=1}^{+\infty} \int_0^t \cdots \int_0^{t_2} \tr(e^{(t-t_k){\cal L}^*_0}{\cal J}^*_{i_k} \cdots {\cal J}^*_{i_k} e^{t_1{\cal L}^*_0}(\rho)) dt_1 \cdots dt_k,\\
\end{split}
\]
and that the expression below safely defines a probability density on $\Omega_t$:
\begin{equation} \label{eq:density}
\frac{d \PP_\rho}{d \mu}(\emptyset)=\tr(e^{t{\cal L}^*_0}(\rho)), \quad \frac{d \PP_\rho}{d \mu}(t_1,i_1,\dots,t_k,i_k)=\tr(e^{(t-t_k){\cal L}^*_0}{\cal J}^*_{i_k} \cdots {\cal J}^*_{i_k} e^{t_1{\cal L}^*_0}(\rho)).
\end{equation}
We will derive a deviation bound for the random variable $N_i(t)$ that counts the number of clicks of the $i$-th detector until time $t$, i.e.
\[N_i(t)(\emptyset)=0, \quad N_i(t)(t_1,i_1,\dots,t_k,i_k)=\sum_{l=1}^k \delta_{i,i_l}.
\]
We recall that the real part of an operator $\eta:\MM \rightarrow \MM$ is defined as $\Re(\eta):=(\eta+\eta^\dagger)/2$.
\begin{theo} \label{thm:ctbound}
Consider an arbitrary, but fixed index $i\in I$ and  define ${\cal A}:=\Re({\cal L})$ and ${\cal B}=\Re({\cal J}_i)$. Assuming that ${\cal A}$ generates an irreducible quantum Markov semigroup, then for any $t \geq 0$, $\gamma>0$ the following inequality holds true:
\begin{equation} \label{eq:ctbound}
    \PP_\rho \left (\frac{N_i(t)}{t}-m_i \geq  \gamma \right) \leq N_\rho \exp \left (-t \left (\frac{\gamma^2}{2 \left ( m+\frac{2 b^2}{\varepsilon}+  \left (\frac{5\alpha } {\varepsilon} \vee  \frac{5}{2}\right )\gamma \right) } \right ) \right )
\end{equation}
where $m:=\tr(L_i^*L_i \sigma)$ is the intensity of $N_i$ at stationarity, $N_\rho:=\|\sigma^{-\frac{1}{2}}\rho \sigma^{-\frac{1}{2}}\|_2$, $\alpha:= \| {\cal B}\|_2$, $b:=\|{\cal B}(\mathbf{1})\|_2$ and $\varepsilon$ is the spectral gap of ${\cal A}$.
\end{theo}

Once again, with a slight abuse of notation, we denote by $\|{\cal B}\|_2$ the operator norm of ${\cal B}$ induced by considering $\MM$ with the KMS-norm. We remark that a sufficient condition for ${\cal A}$ to generate an irreducible quantum Markov semigroup is that $[H,\sigma]=0$. Indeed, both ${\cal L}$ and ${\cal L}^\dagger$ generate irreducible quantum Markov semigroups with faithful invariant state $\sigma$, hence $\sigma$ is a faithful invariant state for ${\cal A}$ too. Moreover, we can easily compute the GKLS form of ${\cal L}^\dagger$ induced by the one of ${\cal L}$ in equation \eqref{eq:GKLS}, which reads
\[{\cal L}^\dagger(x)=i[\sigma^{-\frac{1}{2}} H \sigma^{\frac{1}{2}},x]+\sum_{i \in I} L_i^{\prime*} x L^\prime_i -\frac{1}{2}\{L^{\prime*}_iL^\prime_i,x\}, \quad x \in \MM,
\]
where we omit the exact expression of $L_i^\prime$'s. Since we assumed that $[H,\sigma]=0$, $\sigma^{-\frac{1}{2}} H \sigma^{\frac{1}{2}}=H$. Putting together the GKLS forms of ${\cal L}$ and ${\cal L}^\dagger$, we can express ${\cal A}=({\cal L}+{\cal L}^\dagger)/2$ as
\[\begin{split}
    i[H,x]&+\frac{1}{2} \left (\sum_{i \in I} L^{\prime*}_i x L^\prime_i -\frac{1}{2}\{L^{\prime*}_iL^\prime_i,x\} \right .\\
    &+\left .\sum_{i \in I} L_i^* x L_i -\frac{1}{2}\{L_i^*L_i,x\} \right ), \quad x \in \MM.\\
\end{split}
\]
The irreducibility of a quantum Markov semigroup with faithful invariant state is equivalent to the fact that the commutant of the Hamiltonian, the noise operators and their adjoints is equal to $\mathbb{C} \mathbf{1}$ (\cite[Theorem 7.2]{Wo12}), hence we can conclude that ${\cal A}$ generates an irreducible quantum Markov semigroup: indeed, we have that
\[\{H,L_i,L_i^*, L_i^\prime, L_i^{\prime *}\}^\prime = \{H,L_i,L_i^*\}^\prime \cap \{H,L_i^\prime, L_i^{\prime *}\}^\prime=\mathbb{C}\mathbf{1}.
\]
\begin{proof}
Using the explicit expression of the density of $\PP_\rho$ given in equation \eqref{eq:density}, one can show that the Laplace transform of $N_i(t)$ can be expressed in terms of a smooth perturbation of the Lindblad generator:
\[\mathbb{E}_\rho[e^{u N_i(t)}]=\tr(\rho e^{t{\cal L}_u}(\mathbf{1}))
\]
where ${\cal L}_u(\cdot)={\cal L}(\cdot)+(e^u-1){\cal J}_i(\cdot)$; we refer to Appendix A in \cite{BR21} for a proof using quantum stochastic calculus. Using Lumer-Phillips theorem, we get to
\[\mathbb{E}_\rho[e^{uN_i(t)}] \leq  N_\rho e^{t r(u)} \leq  N_\rho e^{t |r(u)|},
\]
where $r(u)$ is the largest eigenvalue of ${\cal A}_u:=\Re({\cal L}_u)$. Notice that ${\cal A}_u$ is a smooth perturbation of ${\cal A}:=\Re({\cal L})$:
\[{\cal A}_u= {\cal A}+\sum_{k\geq 1} \frac{u^k}{k!} {\cal B}
\]
and ${\cal B}=\Re({\cal J}_i)$. From perturbation theory we get that if ${\cal A}$ is the generator of an irreducible quantum Markov semigroup with spectral gap equal to $\varepsilon$ and if we call $\alpha:= \| {\cal B}\|_2$, we can expand $r(u)=\sum_{k \geq 1}r^{(k)} u^k$ around $u=0$ for $u <(2\alpha \varepsilon^{-1}+1)^{-1}$ and the coefficients $r^{(k)}$'s are provided by the following expression:
\[r^{(k)}=\sum_{p=1}^k \frac{(-1)^p}{p} \sum_{\substack{\nu_1+\dots +\nu_p=k,\, \nu_i\geq 1\\
\mu_1+\dots +\mu_p=p-1, \, \mu_j\geq 0}} \frac{1}{\nu_1! \cdots \nu_p!}\TR({\cal B} S^{(\mu_1)}\cdots {\cal B}S^{(\mu_p)}),
\]
$S^{(0)}=-\ket{\mathbf{1}}\bra{\mathbf{1}}$ and for $\mu\geq 1$, $S^{(\mu)}$ is the $\mu$-th power of
\[S^{(1)}=({\cal A}+\ket{\mathbf{1}}\bra{\mathbf{1}})^{-1}-\ket{\mathbf{1}}\bra{\mathbf{1}}.\]
Notice that $\|S^{(\mu)}\|_2=\varepsilon^{-\mu}$ for $\mu\geq 1$.
This time we get that
\[r^{(1)}=\langle \mathbf{1}, {\cal B}(\mathbf{1}) \rangle =\tr(\sigma L_i^*L_i):=m,\]
which is the intensity of $N_i$ in the stationary regime. For the other terms we need to introduce the notation $b:=\|{\cal B}(\mathbf{1})\|_2$; then we get that
\[|r^{(2)}|\leq \frac{m}{2}+ |\langle {\cal B}(\mathbf{1}), S^{(1)} {\cal B}(\mathbf{1}) \rangle \leq  \frac{m}{2}+ \frac{b^2}{\varepsilon}
\]
and, for $k \geq 3$,
\[|r^{(k)}|\leq \begin{cases}\frac{m}{k!}+ \frac{b^2}{5 \alpha}\left ( \frac{5\alpha}{\varepsilon} \right )^{k-1} & \text{ if } 2\alpha/\varepsilon \geq 1 \\
\frac{m}{k!}+ \frac{2b^2}{ 5 \varepsilon} \left ( \frac{5}{2} \right )^{k-1}& \text{ o.w.}\\ \end{cases},
\]
where we used that $\nu_1! \cdots \nu_p! \geq 2^{k-p}$ and for $p \geq 2$
\[\left |\frac{\TR({\cal B} S^{(\mu_1)}\cdots {\cal B}S^{(\mu_p)})}{\nu_1! \cdots \nu_p!} \right | \leq \frac{b^2}{\alpha} \left (\frac{2 \alpha}{\epsilon} \right )^{p-1}\frac{1}{2^{k-1}} \leq \begin{cases} \frac{b^2}{\alpha} \left (\frac{ \alpha}{\epsilon} \right )^{k-1} & \text{ if } \frac{2 \alpha}{\epsilon} \geq 1 \\ \frac{b^2}{\alpha} \left (\frac{2 \alpha}{\epsilon} \right )\frac{1}{2^{k-1}} & \text{ o.w. } \end{cases}
\]
Wrapping up everything, we obtain that
\begin{equation} \label{eq:estru}
|r(u)| \leq m(e^{u}-1) + \frac{b^2}{\varepsilon}u^2 \left (1-\left (\frac{5\alpha }{\varepsilon} \vee \frac{5}{2}\right )u\right )^{-1}
\end{equation}
Hence, one gets that for every $u>0$
\begin{equation} \label{eq:ctlt}
\begin{split}
    \PP_\rho \left (\frac{N_t}{t}-m \geq \gamma \right) \leq N_\rho \exp \left (-t \left (\gamma u-  m(e^u-u-1)-\frac{b^2}{\varepsilon}u^2 \left (1-\left (\frac{5\alpha }{\varepsilon} \vee \frac{5}{2}\right )u \right )^{-1} \right ) \right )
\end{split}
\end{equation}
Notice that the term $e^{tm(e^u-1)}$ in the r.h.s. of equation \eqref{eq:ctlt} is exactly the Laplace transform of a Poisson process with intensity $m$. The extra terms come from the correlations between the process $N_t$ at different times and the convergence towards the stationary regime. The statement follows  from the same computations as in \cite[Lemma 9]{BQJ18}.
\end{proof}
Below we provide some simple bounds for some of the quantities appearing in inequality \eqref{eq:ctbound}.
\begin{itemize}
    \item Using triangular inequality, we get
    \[b:=\|{\cal B}(\mathbf{1})\|_2\leq \frac{\|L_i^*L_i\|_2+\|(\sigma^{\frac{1}{2}} L^*_i\sigma^{-\frac{1}{2}})^*(\sigma^{\frac{1}{2}} L^*_i\sigma^{-\frac{1}{2}})\|_2}{2}.\]
    Then, we can apply Cauchy-Schwarz inequality:
    \[\|L_i^*L_i\|_2=\tr(\sigma^{\frac{1}{2}}L_i^*L_i \sigma^{\frac{1}{2}}L_i^*L_i)^{\frac{1}{2}} \leq \tr(\sigma (L_i^*L_i)^2)^{\frac{1}{2}} \leq \|L^*_i L_i\|^{\frac{1}{2}} m^{\frac{1}{2}}
    \]
    and
    \[\|(\sigma^{\frac{1}{2}} L^*_i\sigma^{-\frac{1}{2}})^*(\sigma^{\frac{1}{2}} L^*_i\sigma^{-\frac{1}{2}})\|_2\leq \|(\sigma^{\frac{1}{2}} L^*_i\sigma^{-\frac{1}{2}})^*(\sigma^{\frac{1}{2}} L^*_i\sigma^{-\frac{1}{2}})^*\|^{\frac{1}{2}} m^{\frac{1}{2}}.
    \]
    \item Notice that
    \[\alpha:=\|{\cal B}\|_2\leq \|L_i^* \cdot L_i\|_2=\|(\sigma^{\frac{1}{2}} L^*_i\sigma^{-\frac{1}{2}})^*\cdot (\sigma^{\frac{1}{2}} L^*_i\sigma^{-\frac{1}{2}})\|_2.\]
    Let us denote by ${\rm Sp}(\sigma)$ the spectrum of $\sigma$; since
    \[\begin{split}
    &L_i \sigma L^*_i \leq L_i L_i^* \leq \frac{\|L_i L_i^*\|}{\min({\rm Sp}(\sigma))} \sigma,   \\
    &(\sigma^{\frac{1}{2}} L^*_i\sigma^{-\frac{1}{2}}) \sigma (\sigma^{\frac{1}{2}} L^*_i\sigma^{-\frac{1}{2}})^* \leq \frac{\|(\sigma^{\frac{1}{2}}L^*_i\sigma^{-\frac{1}{2}}) (\sigma^{\frac{1}{2}} L^*_i\sigma^{-\frac{1}{2}})^*\| \sigma }{\min({\rm Sp}(\sigma))},
    \end{split}
    \]
    Lemma 1.3 in \cite{OP04} implies that
    \[\|{\cal B}\|_2 \leq \min\{\|L_i L_i^*\|,\|(\sigma^{\frac{1}{2}}L^*_i\sigma^{-\frac{1}{2}}) (\sigma^{\frac{1}{2}} L^*_i\sigma^{-\frac{1}{2}})^*\|\}/\min({\rm Sp}(\sigma)).\]
\end{itemize}
We remark that the proof of Theorem \ref{thm:ctbound} works also for the more general case of the counting processes $Y_j$ defined in \cite{BR21}, which correspond to a change of basis before detection.


\section{Extensions and Applications} \label{sec:e&a}


\subsection{Concentration bounds for fluxes of classical Markov chains} \label{ss:cMC}

Let us consider a classical Markov chain $(X_{n})_{n\geq 1}$ with finite state space $E$ and transition matrix $P=(p_{xy})_{x,y \in E}$ which is irreducible and admits a unique invariant measure $(\sigma_x)_{x \in E}$. Instead of looking at functions of the state of the Markov chain, one may be interested in having concentration bounds for empirical fluxes, i.e. empirical means of functions of the jumps $f:E^2 \rightarrow \mathbb{R}$ (for instance for estimating jump probabilities). Having a wide range of concentration bounds for the empirical mean of functions of the state of a Markov chain, a first natural attempt is considering fluxes as functions of the state of the doubled up Markov chain $(\tilde{X}_n)_{n \geq 1}$, which is the Markov chain with state space $\tilde{E}:=\{(x,y) \in E^2:p_{xy}>0\}$ and with transition matrix given by $\tilde{P}=(\tilde{p}_{(x,y)(z,w)})$, with $\tilde{p}_{(x,y)(z,w)}=\delta_{y,z}p_{yw}$; if $P$ is irreducible, then so is $\tilde{P}$ and its unique invariant distribution is the measure $(\sigma_x p_{xy})_{x,y}$. However, in general the matrix $\tilde{P}$ behaves in a less nice way than $P$: for instance, Theorems 1.1 and 3.3 in \cite{Le98} can never be applied to the doubled up Markov chain for non-trivial models, since both $\tilde{P}$ being selfadjoint or $\tilde{P}^\dagger \tilde{P}$ being irreducible imply that $E$ is a singleton. Remarkably, we can carry out the proofs of Theorems \ref{theo:main1} and \ref{theo:main2} in this classical setting and they provide concentration inequalities for empirical fluxes involving the matrix $P$ instead of $\tilde{P}$: this reflects the fact that $P$ already contains all the information about jumps. Let $\nu$ be any initial probability measure on $E$; for the Bernstein bound it is enough to notice that for every $u>0$
\[\mathbb{E}_\nu[e^{u \sum_{k=1}^{n} f(X_k,X_{k+1})}]=\nu \cdot P_u^{n} \cdot \mathbf{1}
\]
where $(P_u)_{xy}=p_{xy}e^{uf(x,y)}$ and $\mathbf{1}$ is the constant function on $E$ equal to $1$. The proof of Theorem \ref{theo:main1} can be carried out replacing $\MM$ with $\ell^\infty(E):=\{f:E \rightarrow \cc\}$ and $\Phi$ with $P$; in this particular setting $\ell^\infty(E)$ can be turned into a Hilbert space with respect to the inner product
\[\langle h, g \rangle=\sum_{x \in E}\sigma_x \overline{h}(x)g(x), \quad h,g \in \ell^\infty(E)
\]
The Bernstein-type inequality for fluxes reads as follows. We recall that $\tilde{E}:=\{(x,y) \in E^2:p_{xy}>0\}$, $\pi(x,y)=\sigma_xp_{xy}$.
\begin{prop}
    \label{prop1}
    If $Q:=P^\dagger P$ is irreducible, then for every $f:\tilde{E} \rightarrow \rr$ such that $\pi(f)=0$, $\pi(f^2)=b^2 $ and $\| f \|_\infty=c$ for some $b,c>0$ and for every $\gamma > 0$, $n \geq 1$
        \begin{equation}
        \PP_\nu \left (\frac{1}{n} \sum_{k=1}^{n} f(X_k,X_{k+1})\geq\gamma \right) \leq N_\nu \exp \left ( -n \frac{\gamma^2 \varepsilon}{6b^2}h\left ( \frac{10c \gamma}{3b^2}\right )\right )
        \end{equation}
    where $\varepsilon$ is the spectral gap of $Q$, $N_\nu:=\|\frac{d\nu}{d\sigma}\|_2$ and $h(x)=(\sqrt{1+x}+x/2+1)^{-1}$.
\end{prop}

Regarding the Hoeffding-type inequality, one just needs to notice that, calling $a_f \in \ell^\infty(E)$ the unique centered solution of
\[({\rm Id}-P)(h)(x)=\sum_{y \in E}p_{xy}f(x,y), \quad h \in \ell^\infty(E),
\]
for every $k \geq 1$, we can write $f(X_k,X_{k+1})=g(X_k,X_{k+1})-\mathbb{E}_\nu[g(X_k,X_{k+1})|X_{k-1},X_k]$, where $g(x,y)=f(x,y)+a_f(y)$ (once again we assume that $\pi(f)=0$). Repeating the same steps as in the proof of Theorem \ref{theo:main2}, we can obtain the following.
\begin{prop}
    \label{prop2}
    For every $f:\tilde{E} \rightarrow \mathbb{R}$ such that $\pi(f)=0$ and $\|f\|_\infty=c$ for some $c >0$, then for every $\gamma > 0$
        \begin{equation}
            \PP_\nu \left (\frac{1}{n} \sum_{k=1}^{n} f(X_k,X_{k+1})\geq \gamma \right) \leq \exp \left ( - \frac{(n \gamma -2G)^2}{2(n-1)G^2}\right ) \text{ for } n \gamma \geq 2G,
        \end{equation}
    where $G=(1+\|({\rm Id}-P)^{-1}_{|\mathcal{F}}\|_\infty)c$ and $\mathcal{F}:=\{h \in \ell^\infty(E): \sigma(h)=0\}$.
\end{prop}
In the same spirit, one can obtain new bounds for empirical fluxes of continuous time Markov chains too; we aim at studying some applications of such bounds in a future work.
\subsection{Reducible Quantum Channels} \label{sub:redqc}
The study of some physically relevant models requires hypothesis (H) to be dropped: for instance, in the case of non-demolition measurements (\cite{BBB13,BB11}), where the interaction between the system and the ancillas is such that it preserves some non-degenerate observable $A=\sum_{j=1}^d \alpha_j \ket{\alpha_j}\bra{\alpha_j} \in \MM $, hence it is of the form
\[U=\sum_{j=1}^d \ket{\alpha_j}\bra{\alpha_j} \otimes U_j \in M_{d \times \dim(\hh_a)}(\cc)
\]
for some collection of unitary operators $(U_j)$ acting on the ancillary system (which is described by the Hilbert space $\hh_a$). In this case, Kraus operators of $\Phi^*(\rho)=\tr_{\hh_a}(U \rho \otimes \ket{\chi}\bra{\chi} U^*)$ induced by the measurement on the ancilla corresponding to the orthonormal basis $\{\ket{i}\}_{i \in I}$ are given by
\[V_i=\sum_{j=1}^d \bra{i} U_j\ket{ \chi} \ket{\alpha_j}\bra{\alpha_j}, \quad i \in I.
\]
It is easy to see that $\Phi$ is positive recurrent, but it is not irreducible anymore: any $\ket{\alpha_j} \bra{\alpha_j}$ is an invariant state for $\Phi$. If we do not assume that $\Phi$ is irreducible, the first issue we need to take into account is that in general $\frac{1}{n}\sum_{k=1}^{n}f(X_k)$ does converge to a non-trivial random variable: for instance, in the case of non-demolition measurements, one can show that there exists a random variable $\Gamma$ taking values in $\{1,\dots, d\}$ such that
\[\lim_{n \rightarrow +\infty}\frac{1}{n}\sum_{k=1}^{n}f(X_k)=m_\Gamma \quad \PP_\rho \text{-a.s.,}
\]
where $m_j=\sum_{i \in I} f(i) |\bra{i}U_j \ket{\chi}|^2$. What one can do, then, is to upper bound the probability that $\frac{1}{n}\sum_{k=1}^{n}f(X_k)$ deviates from $m_\Gamma$. Leaving aside non-demolition measurements, we will present the result in the case of general quantum Markov chains assuming only the existence of a faithful invariant state for $\Phi^*$. It is well known (\cite{BN12}) that we can always find a (non-unique) decomposition 
\begin{equation} \label{eq:dec}
\cc^d= \bigoplus_{j \in J}\hh_j
\end{equation}
such that
\begin{enumerate}
    \item $\hh_j \perp \hh_{j^\prime}$ if $j \neq j^\prime$ (orthogonality),
    \item if ${\rm supp}(\rho)\subseteq \hh_j$, then ${\rm supp}(\Phi^*(\rho))\subseteq \hh_j$ (invariance),
    \item $\Phi^*$ restricted to $p_{\hh_j}\MM p_{\hh_j}$ is irreducible with unique invariant state $\sigma_j$ (minimality),
\end{enumerate}
where we used the notation $p_{\hh_j}$ to denote the orthogonal projection onto $\hh_j$; one can show that $2.$ is equivalent to
\[2^\prime. \qquad [V_i,p_{\hh_j}]=0 \quad \forall i \in I, j \in J.
\]
If the system initially starts in a state $\rho$ supported on only one of the $\hh_j$'s, it is easy to see that $2.$ and $3.$ imply that we are back to the case of an irreducible quantum channel.
Otherwise, the decomposition in equation \eqref{eq:dec} allows to express $\PP_\rho$ as a convex mixture of probability measures corresponding to irreducible quantum channels:
\[
\begin{split}
    &\PP_\rho(i_1,\dots,i_n)=\tr(V_{i_n} \cdots V_{i_1} \rho V_{i_1}^* \cdots V_{i_n}^*)\\
    &=\sum_{j \in J}\underbrace{\tr(p_{\hh_j} \rho)}_{=:\lambda_j(\rho)} \tr\left (p_{\hh_j}V_{i_n}p_{\hh_j} \cdots p_{\hh_j}V_{i_1}p_{\hh_j} \underbrace{\frac{p_{\hh_j}\rho p_{\hh_j}}{\tr(p_{\hh_j} \rho)}}_{=:\rho_j} (p_{\hh_j}V_{i_1} p_{\hh_j})^*\cdots (p_{\hh_j}V_{i_n}p_{\hh_j})^* \right )\\
    &=\sum_{j \in J} \lambda_j(\rho) \PP_{\rho_j}(i_1,\dots, i_n).
\end{split}
\]
The last result that we need to recall (\cite[Theorem 3.5.2]{Gi22}) is that, given any $f:I \rightarrow \mathbb{R}$, there exists a random variable $\Gamma$ taking values in $J$ such that
\[\lim_{n \rightarrow +\infty}\frac{1}{n}\sum_{k=1}^{n}f(X_k)=\pi_\Gamma(f) \quad \PP_\rho \text{-a.s.,}
\]
where $\pi_j(f)=\sum_{i \in I} f(i) \pi_j(i)$ and $\pi_j(i)=\tr(\sigma_j V_i^*V_i)$. As we already mentioned, under $\PP_{\rho_j}$
\[\lim_{n \rightarrow +\infty}\frac{1}{n}\sum_{k=0}^{n-1}f(X_k)=\pi_j(f) \text{ a.s.,}
\]
which means that ${\rm supp}(\PP_{\rho_j}) \subseteq \{\pi_\Gamma(f)=\pi_j(f)\}$. We have now all the ingredients required to apply previous results in this more general instance: let $\gamma>0$, $n \geq 1$, then
\[\PP_\rho\left ( \left |\frac{1}{n}\sum_{k=1}^{n}f(X_k)-\pi_\Gamma(f) \right |\geq\gamma \right )= \sum_{j \in J}\lambda_j(\rho)\PP_{\rho_j}\left ( \left |\frac{1}{n}\sum_{k=1}^{n}f(X_k)-\pi_j(f) \right |\geq \gamma \right )
\]
and we can apply either Theorem \ref{theo:main1} or Theorem \ref{theo:main2} for upper bounding every $\PP_{\rho_j}\left ( \left |\frac{1}{n}\sum_{k=1}^{n}f(X_k)-\pi_j(f) \right |>\gamma \right )$. We remark that for applying Theorem \ref{theo:main1}, $\Phi^\dagger_{|p_{\hh_j}\MM p_{\hh_j}} \Phi_{|p_{\hh_j}\MM p_{\hh_j}}$ needs to be irreducible, while Theorem \ref{theo:main2} requires that $n\gamma >2(1+\|({\rm Id}-\Phi_{|p_{\hh_j}\MM p_{\hh_j}})_{{\cal F}_j}^{-1}\|_\infty)$ for ${\cal F}_j:=\{x \in p_{\hh_j}\MM p_{\hh_j}: \tr(\sigma_j x)=0\}$ (by $\Phi_{|p_{\hh_j}\MM p_{\hh_j}}$ we mean $(\Phi^*_{|p_{\hh_j}\MM p_{\hh_j}})^*$).


\subsection{Time-dependent and Imperfect Measurements} \label{ss:timedep}
Both inequalities in Theorems \ref{theo:main1} and \ref{theo:main2} can be easily generalized to the setting where we allow the measurement to change along time and the evolution of the system state after the measurement is taken to be more general (including imperfect measurements). Let us consider $\Phi$ an irreducible quantum channel with unique faithful invariant state $\sigma$ and suppose that for any time $n \in \mathbb{N}$ we consider a (possibly) different unravelling of $\Phi$, i.e. a collection of completely positive, sub-unital maps $\{\Phi_{i,n}\}_{i \in I_n}$ such that $|I_n|<+\infty$ and $\Phi=\sum_{i \in I_n} \Phi_{i,n}$; at any time we pick a function $f_n:I_n \rightarrow \mathbb{R}$. In the time inhomogeneous case, even if the system starts from the unique invariant state $\sigma$, the law of the process changes at any time: under $\PP_\sigma$, $X_k$ is distributed according to the probability measure $\pi_k$ defined as follows:
\[\pi_k(i)=\tr(\Phi_{i,k}^*(\sigma)), \quad i \in I_k.
\]
We denote by $\pi_k(f_k)$ the expected value of $f_k$ under $\pi_k$, i.e. $\pi_k(f_k)=\sum_{i \in I_k} f_k(i) \pi_k(i)$.

The proofs of Theorems \ref{theo:main1} and \ref{theo:main2} can be carried out in this more general setting with minimal modifications:
\begin{itemize}
    \item \textbf{Bernstein-type inequality:}
\begin{equation} \label{eq:tdm1}
\PP_\rho \left (\frac{1}{n} \sum_{k=1}^{n}\left ( f_k(X_k)- \pi_k(f) \right ) \geq \gamma \right) \leq N_\rho \exp \left ( - n\frac{\gamma^2 \varepsilon}{6b_n^2}h\left ( \frac{10 c_n \gamma}{3b_n^2}\right )\right )
\end{equation}
where $N_\rho=\|\sigma^{-\frac{1}{2}}\rho \sigma^{-\frac{1}{2}}\|_2$, $\pi_k(f_k):=\sum_{i \in I_k} f_k(i) \tr(\Phi_{i,k}^*(\sigma))$, $c_n:=\max_{k=1,\dots,n}\|f_k-\pi_k(f_k)\|_\infty$ and $b_n^2:=\frac{1}{n} \sum_{k=1}^{n}\pi_k((f_k-\pi_k(f_k))^2)$.

\bigskip Explicit computations show that, in this case, the Laplace transform have the following form:
    \[\mathbb{E}_\rho[e^{u \sum_{k=1}^{n}(f_k(X_k)-\pi_k(f_k))}]=\tr(\rho \Phi_{u,1} \cdots \Phi_{u,n}({\bf 1})) \leq \|\sigma^{-\frac{1}{2}}\rho \sigma^{-\frac{1}{2}}\|_2 \cdot \prod_{k=1}^{n}\|\Phi_{u,k}\|_2,\]
    where $\Phi_{u,k}(x)=\sum_{i \in I_k} e^{u (f_k(i)-\pi_k(i))} \Phi_{i,k}(x)$. One can use the same techniques as in the proof of Theorem \ref{theo:main1} to upper bound every single $\|\Phi_{u,k}\|_2$ getting to the expression below:
    \[\begin{split}
        \mathbb{E}_\rho[e^{u \sum_{k=1}^{n}(f_k(X_k)-\pi_k(f_k))}]& \leq N_\rho e^{ \frac{3u^2}{\varepsilon} \sum_{k=1}^{n} \pi_k((f_k-\pi_k(f_k))^2) \left ( 1-\frac{10 \|f_k-\pi_k(f_k)\|_\infty u}{\varepsilon} \right )^{-1}} \\
        &\leq N_\rho e^{-n\left (-\frac{3u^2}{\varepsilon}  \frac{1}{n}\sum_{k=1}^{n} \pi_k((f_k-\pi_k(f_k))^2)\left ( 1-\frac{10c_n u}{\varepsilon} \right )^{-1} \right )}. \\
        \end{split}\]
Applying Chernoff bound and optimizing on $u>0$, one obtains equation \eqref{eq:tdm1}.
\item \textbf{Hoeffding-type inequality:}
\begin{equation} \label{eq:tdm2}
\PP_\rho \left (\frac{1}{n} \sum_{k=1}^{n}\left ( f_k(X_k)- \pi_k(f) \right ) \geq \gamma \right) \leq \exp \left ( - \frac{(n \gamma -G_n)^2}{(n-1) G_n^2}\right ) \text{ for } n \gamma \geq G_n,
\end{equation}
where $G_n=(1+\sum_{j=0}^{n-2}\|\Phi_{|\mathcal{F}}^j\|_\infty)c_n$ ($\sum_{j=0}^{-1}$ must be interpreted as $0$), $\mathcal{F}:=\{x \in \MM : \tr(\sigma x)=0\}$ and $c_n$ is the same as above.

\bigskip Let us define
\[Z_k^{(n)}=\sum_{j=k}^{n} \mathbb{E}_\rho[f(X_j)|X_1, \rho_1,\dots, X_k, \rho_k].
\]
Notice that for $k=1,\dots, n-1$, $Z_k^{(n)}=f_k(X_k)+\mathbb{E}_\rho[Z_{k+1}^{(n)}|X_1,\rho_1,\dots, X_k,\rho_k]$, hence we can write
\[\sum_{k=1}^{n} f_k(X_k)=\sum_{k=2}^{n}\underbrace{Z_k^{(n)}-\mathbb{E}_\rho[Z_{k}^{(n)}|X_1,\rho_1,\dots, X_{k-1},\rho_{k-1}]}_{D_k}+Z_1^{(n)}.
\]
By Markov property, we have that $Z_k^{(n)}=g_k^{(n)}(X_k,\rho_k)$ and using the explicit expression of the transition operator of $(X_n,\rho_n)$, one gets that $g_k^{(n)}$ has the following form:
\[g_k^{(n)}(i,\omega)=f_k(i)+\sum_{j=0}^{n-k-1} \tr(\omega \Phi^j({\bf F}_{k+j+1})), \quad {\bf F}_{k+j+1}=\sum_{i \in I_{k+j+1}}f_{k+j+1}(i)\Phi_{i,k+j+1}({\bf 1}).
\]
Equation \eqref{eq:tdm2} follows from the same reasoning as in the proof of Theorem \ref{theo:main2} once we observe that
\[\|g_k^{(n)}\|_\infty \leq \left (1+\sum_{j=0}^{n-k-1}\|\Phi_{|\mathcal{F}}^j\|_\infty \right ) c_n \leq \left (1+\sum_{j=0}^{n-2}\|\Phi_{|\mathcal{F}}^j\|_\infty \right ) c_n.
\]

\end{itemize}

It would be interesting to generalize these results to the case where measurements are chosen adaptively, i.e. they may depend on the outcomes of the previous measurements; this would find applications for instance in the task of estimating unknown parameters of the unitary interaction $U$ between the system and the ancillas (see Section \ref{sec:pe}), since an adaptive measurement strategy has been recently shown to be able to asymptotically extract from the output ancillas the maximum amount of information about the unknown parameter (\cite{GG22}).
\subsection{Multi-time Statistics}

In some cases, one is interested in functions of the output measurements at different times: for instance, as in the case of classical Markov chains treated in Section \ref{ss:cMC}, the task could be estimating the rate of jump at stationarity from a certain state to another. Previous techniques still provide bounds for this kind of situations: given $m\geq 2$ and $f:I^m \rightarrow \mathbb{R}$, the natural stochastic process to consider is the one given by $(\underline{X}_n,\rho_n)$ where $\underline{X}_n:=(X_k,\dots,X_{k+m-1})$ and $\rho_n$ is the conditional system state. It is easy to see that it is a Markov process with the following transition probabilities: for every $i_1,\dots, i_m \in I$
\[
\begin{split}
    &\PP_\rho \left (\underline{X}_1=(i_1,\dots,i_m),\rho_1=\frac{V_{i_m}\cdots V_{i_1}\rho V_{i_1}^*\cdots V_{i_m}^*}{\tr(V_{i_m}\cdots V_{i_1}\rho V_{i_1}^*\cdots V_{i_m}^*)} \right )=\tr(V_{i_m}\cdots V_{i_1}\rho V_{i_1}^*\cdots V_{i_m}^*)\\
\end{split}
\]
and for every $n\geq 1$, $j \in I$
\[
\begin{split}
    &\PP_\rho \left (\underline{X}_{n+1}=(i_2,\dots,i_m,j),\rho_{n+1}=\left .\frac{V_{j}\omega V_{j}^*}{\tr(V_{j}\omega V_{j}^*)} \right |\underline{X}_n=(i_1,\dots, i_m), \rho_n=\omega \right )=\tr(V_{j}\omega V_{j}^*).\\
\end{split}
\]
All the other possibilities occur with zero probability.

The transition operator corresponding to this enlarged process is the following
\[P(g)(i_1,\dots,i_m,\omega)=\sum_{j \in I}g\left (i_2,\dots,i_m,j,\frac{V_j\omega V^*_j}{\tr(V_j\omega V^*_j)} \right ) \tr(V_j \omega V_j^*)
\]
where $i_1,\dots,i_m \in I$, $\omega$ is a state on $\hh$ and $g$ is a bounded measurable function. With the same heuristic reasoning used for one-time statistics, we can provide a solution for the Poisson equation in this case too. Let us define the following function:
\[g(\underline{i}, \omega)=f(\underline{i})+\tr(\omega (B^{(m)}(\underline{i})+A^{(m)}_f))
\]
where $\underline{i}=(i_1,\dots,i_m) \in I^m$, $\omega$ is a state on $\hh$,
\[B^{(m)}(\underline{i})=\sum_{k=1}^{m-1} \underbrace{\sum_{j_1,\dots,j_k \in I} f(i_{k+1},\dots, i_{m},j_1,\dots,j_k)V_{j_k}^*\cdots V_{j_1}^*V_{j_1}V_{j_k}}_{=:F_k^{(m)}}
\]
and $A_f^{(m)}$ is the unique solution of
\[(\Id-\Phi)(A_f^{(m)})=F_m^{(m)}, \quad F_m^{(m)}:=\sum_{j_1,\dots,j_m \in I} f(j_1,\dots,j_m)V_{j_m}^*\cdots V_{j_1}^*V_{j_1}V_{j_m}
\]
such that $A_f^{(m)} \in {\cal F}=\{x \in \MM : \tr(\sigma x)=0\}$. Notice that $A_f^{(m)}$ exists if
\[\tr(\sigma F_m^{(m)})=\sum_{j_1,\dots,j_m \in I} f(j_1,\dots,j_m)\tr(\sigma V_{j_m}^*\cdots V_{j_1}^*V_{j_1}V_{j_m})=0,
\]
which means that the function $f$ is centered with respect to the probability measure $\pi^{(m)}(j_1,\dots, j_m)=\tr(\sigma V_{j_m}^*\cdots V_{j_1}^*V_{j_1}V_{j_m})$ (which is the law of $m$ consecutive measurements in the stationary regime). Reasoning as for $m=1$, we get the following inequality.
\begin{prop}
For every $m\geq 1$, $f:I^m \rightarrow \mathbb{R}$ such that
\[\pi^{(m)}(f)=\sum_{j_1,\dots, j_m \in I} f(j_1,\dots, j_m) \tr(\sigma V_{j_1}^* \cdots V_{j_m}^* V_{j_m} \cdots V_{j_1})=0\]
and $\|f\|_\infty=c$ for some $c >0$, then for every $\gamma > 0$
\begin{equation}
\PP_\rho \left (\frac{1}{n} \sum_{k=1}^{n} f(\underline{X}_k)\geq \gamma \right) \leq \exp \left ( - \frac{(n \gamma -2G)^2}{2 (n-1) G^2}\right ) \text{ for } n \gamma \geq 2G,
\end{equation}
where $G=(m+\|({\rm Id}-\Phi)^{-1}_{|\mathcal{F}}\|_\infty)c$ and $\mathcal{F}:=\{x \in \MM : \tr(\sigma x)=0\}$.
\end{prop}


\subsection{Parameter Estimation} \label{sec:pe}

Concentration inequalities in Theorems \ref{theo:main1} and \ref{theo:main2} can be used in order to find confidence intervals for dynamical parameters (and possibly perform hypothesis testing): suppose that the unitary interaction  $U$ between the system and the ancillas depends on a unknown parameter $\theta \in \Theta \subseteq \mathbb{R}$, which we want to estimate via indirect measurements. Kraus operators $V_i(\theta)$ and the steady state $\sigma(\theta)$ depend on the parameter too and so does the asymptotic mean:
\[\pi(f)(\theta)=\sum_{i \in I} f(i)
\tr(\sigma(\theta)V_i(\theta)^*V_i(\theta)).\]
For the sake of clarity, we are now going to treat the simplest case in which $\overline{f}_n:=\frac{1}{n}\sum_{k=1}^{n}f(X_k)$ is a consistent estimator for $\theta$, i.e. when $\pi(f)(\theta)=\theta$; however, one can easily generalize the same reasoning to more general instances. Theorems \ref{theo:main1} and \ref{theo:main2} can  be used to estimate the probability that $\theta$ lays in a interval centered at $\overline{f}_n$: indeed, for every $\gamma>0$
\[\begin{split}
    &\PP_{\rho,\theta} \left (\theta \in \left ( f_n - \gamma, f_n+ \gamma \right )\right ) =\PP_{\rho,\theta} \left (\left |f_n-\theta \right | < \gamma \right)=1-\PP_{\rho,\theta} \left (\left |f_n-\theta \right | \geq \gamma \right)\\
    & \geq 1-2N_\rho \exp \left (-\max \left \{ n\frac{ \gamma^2 \varepsilon(\theta)}{8b^2(\theta)}h \left ( \frac{5c \gamma}{b^2(\theta)}\right ),  \frac{(n \gamma -2G(\theta))^2}{2(n-1)G(\theta)^2} 1_{\{n\gamma >2G(\theta)\}}\right \} \right). 
\end{split}
\]
\sloppy Since the real value of the parameter $\theta$ is unknown, one is usually interested in lower bounding $\PP_{\rho,\theta} \left (\theta \in \left ( f_n - \gamma, f_n+ \gamma \right ) \right )$ uniformly for $\theta \in \Theta$ or its average with respect to a certain prior measure $\mu$ on $\Theta$, i.e. $\int_{\Theta}\PP_{\rho,\theta} \left (\theta \in \left ( f_n - \gamma, f_n+ \gamma \right ) \right )d\mu(\theta)$; in order to obtain meaningful lower bounds, in the first case one needs either $1/\varepsilon(\theta)$ or $G(\theta)$ to be uniformly bounded, while in the second case it is enough that the set where $1/\varepsilon(\theta)$ or $G(\theta)$ grows unboundedly is given a small probability by the prior $\mu$. We remark that $1/\varepsilon(\theta)$ and $G(\theta)$ approaching $+\infty$ is a phenomenon related to $\Phi_\theta(\cdot):= \sum_{i \in I} V_i^*(\theta) \cdot V_i(\theta)$ losing ergodicity and approaching a phase transition.

\section{Conclusions and Outlook} \label{sec:conclusions}

We derived a generalisation of Bernstein's and Hoeffding's concentration bounds for the time average of measurement outcomes of discrete time quantum Markov chains. 

Our results hold under quite general and easily verifiable assumptions and depend on simple and intuitive quantities related to the quantum Markov chain. We were also able to apply the same techniques employed for showing the Bernstein-type inequality to provide a concentration bound for the counting process of a continuous time quantum Markov process; this result complements deviation bounds obtained using different techniques in \cite{BCJP21}. While our strategy was inspired by works on concentration bounds for the empirical mean for classical Markov chains \cite{Le98,GO02,BQJ18,FJQ21}, when restricted to the classical setting, our results provide extensions to empirical fluxes of the corresponding classical bounds. 

Our work here finds a natural application in the study of finite-time fluctuations of dynamical quantities in physical systems, something of core interest in both classical and quantum statistical mechanics. Our work provides the tools to deal with problems which are tackled using concentration bounds in statistical models which involve the more general class of stochastic processes that can be seen as output processes of quantum Markov chains (which include important examples, e.g. independent random variables, Markov chains, hidden Markov models). 

Our results should be useful in several areas. One is the estimation of dynamical parameters in quantum Markov evolutions, where a natural extension of our results would be to the case where measurements are chosen adaptively in time \cite{GG22}, and to more general additive functionals of the measurement trajectory. A second area of interest is in the connection to so-called thermodynamic uncertainty relations (TURs), which are general {\em lower bounds} on the size of fluctuations in trajectory observables such as time-integrated currents or dynamical activities. TURs were postulated initially for classical continuous-time Markov chains \cite{Barato2015b} (and proven via large deviation methods \cite{Gingrich2016}), and later generalised in various directions, including finite time \cite{Pietzonka2017}, discrete Markov dynamics \cite{Proesmans2017}, first-passage times \cite{Garrahan2017,Gingrich2017}, and quantum Markov processes
Refs.~\cite{Brandner2018,Carollo2019,Guarnieri2019,Hasegawa2020}. The concentration bounds like the ones we consider here bound the size of fluctuation {\em from above}, and are therefore complementary to TURs. 
It will be interesting to see how to use the concentration bounds for fluxes obtained here to formulate ``inverse TURs'' that upper-bound dynamical fluctuations in terms of general quantities of interest like entropy production and dynamical activity, as happens with standard TURs.

\bigskip \noindent \paragraph{Aknowledgements.} This work was supported by the EPSRC grant EP/T022140/1. 

\noindent \paragraph{Declarations.}

\noindent \paragraph{Data Availability Statement.} Data sharing not applicable to this article as no datasets were generated or analysed during the current study.

\bibliographystyle{abbrv}
\bibliography{biblio}

\end{document}